\documentclass[msom,nonblindrev]{informs3}

\DoubleSpacedXI

\usepackage{natbib}
 \bibpunct[, ]{(}{)}{,}{a}{}{,}%
 \def\bibfont{\small}%

\usepackage[utf8]{inputenc} 
\usepackage[T1]{fontenc}    
\usepackage{hyperref}       
\usepackage{url}            
\usepackage{booktabs}       
\usepackage{amsfonts}       
\usepackage{nicefrac}       
\usepackage{microtype} 
\usepackage{enumerate}
\usepackage{comment} 
\newcommand{\N}{\mathbb{N}}

\newcommand{\I}{\mathbb{I}}
\newcommand{\R}{\mathbb{R}}

\def\BD{D}
\def\BDi{D_{i}}

\def\BDone{D_{1}}
\def\BDtwo{D_{2}}
\def\BDn{D_{n}}
\def\trho{t_{\theta}}
\def\const{c}
\def\V{c_2}
\def\te{t_{\epsilon}}

\def\teta{t_{\eta}}

\def\ourrho{\bar \theta}
\def\NN{\rho_a}
\def\MN{t_{a}}
\def\Q{u}
\def\tildet{\widetilde t}

\usepackage{algorithm}
\usepackage{algpseudocode}
\usepackage{amsmath,amssymb,amsfonts}
\usepackage{enumitem}
\usepackage[dvipsnames]{xcolor}
\usepackage{ mathrsfs }
\usepackage{bm}
\usepackage{thmtools,thm-restate}
\usepackage{cleveref}
\usepackage{subcaption}
\usepackage{graphicx}

\newcommand{\negin}[1]{\noindent{\textcolor{red}{\{{\bf NG:} #1\}}}}
\newcommand{\jasoncom}[1]{\noindent{\textcolor{blue}{\{{\bf JL comment:} #1\}}}}

\def\low {\underline p}
\def\high{\bar p}
\def\margin{m}

\def\cone{s}

\def\var{z}
\newcommand{\norm}[1]{\lVert#1\rVert}

\newcommand{\defeq}{:=}
\newcommand{\cond}[1]{\left. \right|}

\newtheorem{theorem}{Theorem}[section]
\newtheorem{corollary}{Corollary}[theorem]
\newtheorem{lemma}[theorem]{Lemma}
\newtheorem{proposition}[theorem]{Proposition}
\newtheorem{definition}{Definition}[section]
\newtheorem{example}{Example}
\newtheorem{assumption}{Assumption}

\begin{document}




\TITLE{No-regret Learning in Price Competitions under Consumer Reference Effects}

\ARTICLEAUTHORS{%
\AUTHOR{Negin Golrezaei}
\AFF{Sloan School of Management, Massachusetts Institute of Technology \EMAIL{golrezaei@mit.edu}, \URL{}}
\AUTHOR{Patrick Jaillet}
\AFF{Department of Electrical Engineering and Computer Science, Massachusetts Institute of Technology \EMAIL{jaillet@mit.edu}, \URL{}}
\AUTHOR{Jason Cheuk Nam Liang}
\AFF{Operations Research Center, Massachusetts Institute of Technology \EMAIL{jcnliang@mit.edu}, \URL{}}
} 

\ABSTRACT{
We study long-run market stability for repeated price competitions between two firms, where consumer demand depends on firms' posted prices and consumers’ price expectations called \emph{reference prices}. Consumers' reference prices vary over time according to a {memory-based} dynamic, which is a weighted average of all historical prices. We focus on the setting where firms are not aware of demand functions and how reference prices are formed but have access to an oracle that provides a measure of consumers'
responsiveness to the current posted prices. We show that if the firms run no-regret algorithms, in particular, online mirror descent (OMD), with decreasing step sizes,  the market stabilizes in the sense that firms' prices and reference prices converge to a stable Nash Equilibrium (SNE). Interestingly, we also show that there exist constant step sizes under which the market stabilizes. {We further characterize the rate of convergence to the SNE for both decreasing and constant OMD step sizes.}}%


\maketitle

%


\section{Introduction}
\label{sec:intro}
In markets with repeated consumer-seller interactions, consumers develop price expectations (or reference prices) based on past observed prices. Such price memories would influence consumers' willingness-to-pay and hence their purchasing decisions, eventually impacting the overall aggregate market demand. {Due to such memory dependent reference price effects, developing pricing strategies is challenging because firms may not necessarily know how consumers form and adjust price expectations.} The complexity of pricing is further increased with competition, as competitors' pricing decisions impact not only a firm's immediate demand but also  consumers' reference prices. 
Such challenges in pricing  under competition and  reference price effects make market stability particularly attractive to firms:  under stable markets, 
long-term organizational planning and business strategy development can be conducted more effectively  (see \cite{caves1978market}). 
Inspired by this, in this paper, \emph{we study the impact of consumer reference prices on the long-term stability of competitive markets.}

We examine a simplified market scenario where two firms sequentially set prices to sell goods over an infinite time horizon, and demand of each firm's goods are  influenced by both firms' current prices and the current consumers'  reference price, which is a weighted average of all past price trajectories. Also, the repeated price competitions occur in an opaque environment, where firms are not aware of any demand or reference price characteristics, and only have access to an oracle that returns consumers' responsiveness to posted prices.\footnote{We consider a linear demand model, but firms are not aware of the functional form of the demand.} 
In such a market scenario, we consider that both firms run a general online mirror descent (OMD) algorithm.\footnote{OMD algorithms are closely related to the regularized learning paradigm, which includes algorithms such as follow the regularized leader (FTRL), EXP3, Hedge, etc (see \cite{hoi2018online} for a comprehensive survey).} {Despite its simplicity, OMD algorithms have been theoretically shown to have good performance guaranties in both purely stochastic and adversarial environments (see \cite{bubeck2012best,zimmert2018optimal,zimmert2019beating}), and hence would be a plausible option for firms in this opaque environment of interest.} 

{Our goal is to investigate  
whether firms' prices and consumer reference prices eventually stabilize in the long-run if firms run OMD.} The notion of stability that we consider is represented by the convergence of firms' price profiles and reference prices such that there is no incentive for firms to deviate, eliminating the possibility for long-run price cycles and fluctuations. Similar notions of stability under dynamic competition has been studied under various equilibrium frameworks, and most relevant to this work are Markov perfect equilibrium and stationary equilibrium (see for example \cite{hopenhayn1992entry,escobar2007existence,doraszelski2010computable,weintraub2011industry, adlakha2015equilibria}). Nevertheless, these frameworks assume firms have complete information or optimize pricing decisions according to some prior on their competitors (or their aggregate) and the market. 
In contrast, our work focuses on competition in an opaque environment where firms post prices using no-regret learning algorithms like OMD. {Here, we point out that our objective is not to present dynamic pricing polices that maximize firms' cumulative revenue. Instead, we seek to shed light on whether simple pricing polices like OMD that do not require a large amount of information eventually achieve market stability. Our contributions are summarized as follows.} 

\begin{itemize}[leftmargin=*]
    \item We characterize stability for dynamic competitive markets under consumers' reference price effects by defining the notion of \textit{Stable Nash Equilibrium (SNE)}.  We theoretically demonstrate its existence and shed light on its structural properties; see Theorem \ref{thm:NEconvergence}.
    \item We transform the two-firm game with a  dynamic state (reference price) that  varies in time according to firms' posted prices to a three-firm game without a state. The added virtual firm, which is referred to as nature, runs OMD with a constant step size (i.e. has a fast learning rate), and models how reference prices are affected by firms' past pricing decisions. 
    
    
    \item We show that prices and reference prices converge to an SNE and achieve stable markets when the two (real) firms adopt decreasing step sizes that go to zero at a moderate rate; see Theorem \ref{thm:OMD:convergenceSufDecStep} for details.  We further show that with decreasing step sizes, the market stabilizes at a linear rate. 
We highlight that obtaining these convergence results is challenging because in our three-firm game, there is a firm (nature) who adopts a constant step size and learns at a fast rate. Our results show that despite the need to deal with such an inflexible virtual firm, the real firms can stabilize the market by adopting  decreasing step sizes. In fact, the existence of the inflexible virtual firm in our game does not allow us to use the results in the literature on multi-agent online learning, where multiple interacting agents make sequential decisions via running the OMD algorithm  to maximize individual rewards (see  \cite{mertikopoulos2017convergence,bravo2018bandit,mertikopoulos2019learning}). More specifically, in the multi-agent online learning literature, agents in the system of interest typically use step sizes of the same order (i.e. homogeneously decreasing or constant). In contrast, in our setting, firms are unaware of reference price update dynamics, and may possibly take decreasing step sizes while nature's step sizes remain constant.
    \item Interestingly, we also show that there exist constant step sizes under which markets will converge to an SNE at much faster rates compared to adopting decreasing step sizes. Additionally, we show through an example that not every constant step size results in a stable market. Roughly speaking, if the firms' constant step size is compatible with nature's constant step size, the market stabilizes at a faster rate compared to decreasing step sizes; see Corollary \ref{cor:OMD:convergenceSuf} and Theorem \ref{thm:OMD:constExpConvergence} for details.   
    
\end{itemize} 
{We refer the readers to  Appendix \ref{app:litreview} for  an expanded literature review.}

\section{Preliminaries} \label{Sec:model}
\paragraph{Consumer Demand and Reference Price Update Dynamics.} We study a dynamic  system where two firms simultaneously set prices in each period over an infinite time horizon to sell goods to consumers whose willingness-to-pay is affected by their price expectations, referred to as \emph{reference prices}. We assume that the number of consumers is large so that  demand for each firm  is governed by the aggregate behavior of all consumers. Specifically, the demand of firm $i\in \{1,2\}$ in time period $t$ with posted prices $\bm{p}_t=(p_{1,t}, p_{2,t})$ and  consumers' reference price $r_{t}$ is given by
\begin{align}
\label{def:RAdemand}
    d_i(p_{i,t}, p_{-i,t}, r_{t}) = \alpha_{i} - \beta_{i} p_{i,t} + \delta_{i} p_{-i,t} + \gamma_{i} r_{t}\,, 
\end{align}
where $p_{i,t}$ is the  price of firm $i$ and $p_{-i,t}$ is the price of the other firm. To simplify notation, we may denote $d_i(p_{i,t}, p_{-i,t}, r_{t})$ with $d_i(\bm{p}_t, r_{t})$. We assume  prices $p_{i,t}$ and reference prices $r_{t}$ are bounded, i.e., for $i\in \{1,2\}$,  $ p_{i,t}, r_{t}\in \mathcal{P} = [\low, \high]$ for some $0 < \low <  \high< \infty$, and $d_{i}(p_{i}, p_{-i}, r) \geq 0$  for any $p_{i}, p_{-i}, r \in \mathcal{P}$. 
The boundedness of prices corresponds to real-world price floors or price caps and is not unnatural. In Equation \eqref{def:RAdemand},  $\alpha_{i}, \delta_{i}, \gamma_{i}> 0$ and  {$\beta_{i} \geq \margin\left(\delta_{1} + \delta_{2} + \max\{\gamma_{1},\gamma_{2}\}\right)$}, where $\margin>0$. Later in this section, we will  provide an  interpretation for these parameters that characterize our linear demand model. We note that linear demand models, which are widely used in the literature (see \cite{huang2013demand} for a comprehensive survey),  can be viewed as a first-order approximation to more complex models. 

After firms post prices, reference prices update according to the following dynamics: 
\begin{align}
\label{eq:referenceUpdate}
    r_{t+1} = ar_{t}  + (1-a)\left(\theta_1 p_{1,t} + \theta_2 p_{2,t} \right)\,,
\end{align}
{where} $\theta_1,\theta_2, a \in (0,1)$ and $\theta_1 + \theta_2 =1$. {Here, $\theta_{i}$, which is independent of prices,  represents how visible firm $i$ is to consumers: the larger the $\theta_{i}$, the more visible firm $i$ is, and the more it influences  consumers' price expectations. 
} The reference price update dynamics can be viewed as a memory-based  process that characterizes how consumers adjust price expectations for goods over time as they observe new prices. Reference prices are formed by a weighted average of historical prices,  where more recent prices are assigned larger weights. The specific exponential weighting scheme adopted in this paper has been motivated and empirically validated in the literature of  behavioral economics (see, for example \cite{winer1986reference,sorger1988reference,greenleaf1995impact}). The parameter $a$ in the reference price update model  characterizes to what extent  consumers' reference price depends on past prices: As $a$ increases,  the reference prices depend less on recently observed prices. 
{Empirical estimates of $a$ typically range from $0.47$ to $0.925$ (see \cite{greenleaf1995impact,briesch1997comparative}) depending on the type of goods sold.}

We now provide an economic interpretation for our linear demand model  by rearranging terms:
\begin{align} \textstyle d_i(p_{i,t}, p_{-i,t},r_{t}) = \alpha_{i} - \left(\beta_{i} - \gamma_{i} \right) p_{i,t} + \delta_{i} p_{-i,t} + \gamma_{i} {\left(r_{t} - p_{i,t}\right)}\,.
\label{eq:d_2}\end{align}
When the posted price is greater than the reference price, i.e., $p_{i,t} > r_{t}$, the value $ p_{i,t} - r_{t}$ can be viewed as the consumers' perceived price surcharge w.r.t. the reference price, and  when $p_{i,t} < r_{t}$,  the value $ r_{t} - p_{i,t} $ is consumers' perceived price discount.
Observe that in this rearrangement, demand increases when consumers' perceived price discount $(r_{t}-p_{i,t})\I\{r_{t}> p_{i,t}\}$
increases, and decreases as price surcharge $(p_{i,t}-r_{t})\I\{p_{i,t} > r_{t}\}$ increases, which is a conventional representation of how  reference prices  affect consumer decisions in the related literature, see, for example, \cite{popescu2007dynamic, nasiry2011dynamic}. Furthermore, the coefficients $\beta_{i} - \gamma_{i}, \delta_{i}$, and  $\gamma_{i}$ measure the demand sensitivity of firm $i$ to its own prices $p_{i,t}$, its competitor's prices $p_{-i,t}$, and price surcharge/discount respectively.\footnote{The dependency of demand on price surcharges and discounts are of the same order $\gamma_{i}$, which corresponds to so-called risk-neutral consumers. Related literature have also studied asymmetric demand dependencies on surcharges and discounts; see \cite{popescu2007dynamic,nasiry2011dynamic,hu2016dynamic}.} With these interpretations, parameter $\margin>0$ in the condition of 
{$\beta_{i} \geq \margin\left(\delta_{1} + \delta_{2} + \max\{\gamma_{1},\gamma_{2}\}\right)$} can be viewed as a \emph{sensitivity margin} that represents to what extent demand is more sensitive to a firm's own prices relative to competitor's prices and surcharge/discount. Take for example the case where $\margin = 1$: we have $\beta_{i} - \gamma_{i} > \delta_{i}$ {$ + \delta_{-i}$}, which means the impact of firm $i$'s price on its demand is greater than the aggregate impact of its price on the competitor's demand and the competitor's price on firm $i$'s demand (see Equation \eqref{eq:d_2}).  
 Additionally, for $\margin = 2$, we  have  $ \max\{\gamma_{1},\gamma_{2}\} < \beta_{i} - \gamma_{i} $,  which represents the fact that reference effects in the market due to  surcharge/discounts are generally less influential than any firm's price on its own demand.

We point out that the aforementioned relationships between model parameters $\{\alpha_{i},\beta_{i},\delta_{i}, \gamma_{i},\theta_{i}\}_{i=1,2}$ lead to  a diagonally dominant Jacobian matrix w.r.t. some mapping that characterizes the linear system consisting of firms and reference prices (particularly linearity in demand and reference price updates). We will provide further details on this particular mapping and its relevance with variational inequalities in Section \ref{sec:comparison}.

\textbf{Market Stability.} In this work, our goal is to present simple pricing policies for the firms that stabilize the market even when firms do not have complete information on market conditions. Define  $\pi_i(\bm{p},r) := p_{i}\cdot  d_i(\bm{p},r)$ as the single-period firm $i$'s revenue when prices are $\bm{p}$ and the reference price is $r$.  
We say the market is stable at point $(\bm{p}^*, r^*)$ if the following two conditions hold:
\begin{enumerate} [leftmargin=*]
    \vspace{-0.2cm}
    \item \textbf{Best-response Conditions.} for $i\in\{1, 2\}$, we have  $ \pi_{i}(p_{i}^{*},{p}_{-i}^{*}, r^{*}) \geq \pi_{i}(p,{p}_{-i}^{*}, r^{*})$ for any $p$ in the feasible set $\mathcal{P}$; that is, firm $i$ cannot increase its revenue by posting another price $p\ne p_i^*$ when the other firm posts a price of ${p}_{-i}^{*}$ and the reference price is $r^*$. 
    \item \textbf{Stability Condition. } $r^{*}= \theta_{1}p_{1}^{*} + \theta_{2}p_{2}^{*}$; that is,
    the reference price does not change if the firm $i\in\{1,2\}$ keeps posting price $p_i^*$; see Equation \eqref{eq:referenceUpdate}.
\end{enumerate}
Throughout the paper, we may refer to a point $(\bm{p}^*, r^*)$ that satisfies the aforementioned conditions as a Stable Nash Equilibrium (SNE). 
\paragraph{Firms' Information Structure.}

We present pricing policies under a {partial information} setting. 

In this setting, a firm $i$ does not know 
$d_i$, $d_{-i}$, reference price update dynamics, and does not observe any of historical competing prices nor the current reference price. To be more specific, in this setting, firms do not know the specific form of the demand functions and reference update dynamics, which in our case are linear.  Nevertheless, we assume that after firms post prices $\bm{p}_{t}$ under reference price $r_{t}$, they can access a first-order oracle that outputs ${\partial 
{\pi}_i(\bm{p}_{t}, r_{t})}/{\partial p_{i}}$, which intuitively represents consumers' responsiveness to a firm's prices under current market conditions.\footnote{We note that such information can be obtained by a slight perturbation of the posted price. Furthermore,  the assumption of having access to the first-order oracle is very common in the literature; see, for example, a comprehensive introduction to convex optimization in \cite{nesterov2013introductory}.} We note that  the partial information setting models real-world opaque environments where firms do not  possess information of the market or its competitors.
In this setting, firms set prices simultaneously, so a firm does not observe its competitor's pricing decision in the current period before setting its own price.

\section{Existence and Structural Properties of SNE}
\label{sec:NE}
In this section, we show that an SNE exists. Recall that for any SNE, each firm best responds to its competitor as well as consumers' reference price with no incentive for unilateral deviation.  Let $\psi_i(p_{-i},r) = \argmax_{p\in \mathcal{P}} \pi_i(p,p_{-i},r)$, $i\in \{1,2\},$\footnote{Here, the revenue function $\pi_{i}$ is quadratic, so $\argmax_{p\in \mathcal{P}} \pi_{i}$ is a singleton.} be firm $i$'s best-response to the reference price $r$ and the price of the other firm $p_{-i}$. Further, 
for any reference price $r$, define 
set  $\mathcal{B}(r)$ as follows
\begin{align}
\label{eq:NE:bestresponseprfile}
\mathcal{B}(r) = \left\{\bm{p}~:~ p_i = \psi_i(p_{-i},r), i = 1,2 \right\}\,. 
\end{align}
As we will show in Theorem \ref{thm:NEconvergence} below, $\mathcal{B}(r)$ is non-empty and when it is not a singleton, it is an ordered set with total ordering.\footnote{A set $\mathcal{A} \subset \R^{d}$ is an ordered set with total ordering if for any $\bm{x},\bm{y} \in \mathcal{A}$, either $\bm{x} \leq \bm{y}$ or $\bm{y} \leq \bm{x}$ where the relationship $\leq$ and 
$\geq$ between two vectors is component-wise.} To show the existence of an SNE, we consider a simple pricing strategy that works as follows: in each period, firms set the largest best response profiles $\bm{p}_{t}$ w.r.t. reference price $r_{t}$, i.e., $\bm{p}_{t} = \max \mathcal{B}(r_{t})$ (because $\mathcal{B}(\cdot)$ is an ordered set, $\max \mathcal{B}(r_{t})$ is well-defined).
We show that  for any initial reference price $r_{1} \in \mathcal{P}$, $(\bm{p}_{t},r_{t})$ converges monotonically to an SNE. Of course, this pricing strategy is only possible under the complete information setting, where each firm knows its own demand function $d_i$, its competitor's demand function $d_{-i}$, and the current reference price.  That is, the described pricing strategy cannot be implemented in our partial information setting. 
Nevertheless, the convergence under this policy confirms the existence of an SNE. 

\begin{theorem}[Existence of an SNE]
\label{thm:NEconvergence}
Let $\mathcal{B}(r)$, defined in Equation \eqref{eq:NE:bestresponseprfile}, be the set of best-response profiles w.r.t. reference price $r$. Then, for a fixed reference price $r \in \mathcal{P}$, $\mathcal{B}(r)$ is non-empty, and when  $\mathcal{B}(r)$ is not a singleton, it is an ordered set with total ordering. Furthermore, assume that in each period $t$, firms set the largest best response prices $\bm{p}_{t}$ w.r.t. reference price $r_{t}$, i.e., $\bm{p}_{t} = \max \mathcal{B}(r_{t})$. Then, for any initial reference price $r_{1} \in \mathcal{P}$, $(\bm{p}_{t},r_{t})$ converges monotonically to an SNE.
\end{theorem}
The proof of the first half of the result regarding the structural properties of the set of best response profiles $\mathcal{B}(r)$ is inspired by that of Tarski's fixed point theorem (e.g., see \cite{echenique2005short}). The proof of the second half regarding the convergence of the pricing policy  builds on that of Theorem 6 in \cite{milgrom1990rationalizability}. (This theorem shows  the monotonocity of pure-strategy Nash Equilibrium for paramtererized games.) Detailed proofs can be found in Appendix \ref{app:NE}.
Theorem \ref{thm:NEconvergence} illustrates structural properties of SNEs: since $\mathcal{B}(\cdot)$ is an ordered set with total ordering, if there are multiple SNE's, any two SNE's $(\bm{p}_{a}^*, r_{a}^*)$ and $(\bm{p}_{b}^{*}, r_{b}^*)$ must either satisfy $\bm{p}_{a}^* \geq \bm{p}_{b}^*$ or $\bm{p}_{a}^* \leq \bm{p}_{b}^*$ under component-wise comparisons.

Due to the decision set boundaries, there may exist multiple SNE's. 
However, to simplify our analyses, in the rest of the paper we assume that there exists an SNE that lies within the interior of the action set $\mathcal{P}$.
Under this assumption, Lemma \ref{lem:uniqueNE} shows that the interior SNE is unique. 
\begin{assumption}
    \label{assum:FOCinterior} There exists an SNE $(\bm{p}^{*},r^{*})$ such that $(\bm{p}^{*},r^{*}) \in (\low, \high)^{3}$.
\end{assumption}
\begin{lemma}[Uniqueness of SNE]
    \label{lem:uniqueNE}
    Under Assumption \ref{assum:FOCinterior}, there is a unique SNE $(\bm{p} ^*, r^*)\in (\low, \high)^{3}$. 
\end{lemma}

\section{No-regret Pricing Policies  under Partial Information Setting}
\label{sec:OMD}


Recall that under partial information, firms 
are unaware of the consumer demand function (they do not know the demand function is linear), reference prices, and reference price update dynamics. Hence, a natural approach for firms to increase revenue is to employ so-called \textit{no-regret online learning} algorithms that adjusts prices in a dynamic fashion.  
We study the regime in which firms adopt the general OMD algorithm. We start by the following standard definition.  

\begin{definition}[Strong convexity] 
Let $\mathcal{C} \subset \R$ be a convex set. A function $R: \mathcal{C} \to \R$ is said to be $\sigma$-strongly convex if for any $x,y \in \mathcal{C}$, we have
$ R(x) - R(y) \geq  \frac{d R(y)}{dy} (x-y) + \frac{\sigma^2}{2}(y-x)^{2}$. 
\end{definition}

In the OMD algorithm, each firm $i$ chooses a continuously differentiable and strongly convex regularizer $R_i: \R \to \R$ associated with strong-convexity parameter $\sigma_i$, a sequence of step sizes $\{\epsilon_{i,t}\}_t$, and, for our convenience, minimizes the cost function (i.e. inverse of revenue) $\widetilde{\pi}_i \defeq -\pi_i$, which is convex in $p_i$. Here, we assume each regularizer also satisfies a standard \textit{reciprocity condition} used in optimization and online learning literature \cite{chen1993convergence,kiwiel1997free,alvarez2004hessian}, i.e. whenever $x\to y$ for $x,y \in \R$ we have $\BD_{i}(x,y)\to 0$ where $\BD_{i}$ is the Bregman divergence w.r.t. $R_{i}$.\footnote{The Bregman divergence $\BD: \mathcal{C} \times \mathcal{C} \to \R^+$ associated with convex and continuously differentiable regularizer function $R: \mathcal{C} \to \R$ and convex set $\mathcal{C} \subset \R$ is defined as $ \BD(x,y) \defeq R(x) - R(y) - R'(y) (x-y)$.
} In OMD, each firm $i$ maintains a \emph{proxy variable} $y_{i,t} \in \R$ over time, and in each period $t$, conducts pricing according to the following three steps: 
\begin{enumerate}[leftmargin=*]
    \item Project the proxy variable $y_{i,t}$ back to the decision interval $\mathcal{P} = [\underline{p},\Bar{p}]$: $p_{i,t} = \Pi_{\mathcal{P}}( y_{i,t})$, where $\Pi_{\mathcal{P}}:\R \to \mathcal{P}$ is the projection operator such that $\Pi_{\mathcal{P}}(z) = z \I\{z \in \mathcal{P}\} + \underline{p} \I\{z < \underline{p}\}+\Bar{p} \I\{z > \Bar{p}\}$.
    \item Access the first-order oracle $g_{i,t} \defeq g_{i}(\bm{p}_{t},r_{t})$ defined by  $g_{i}: \mathcal{P}^{3} \to \R$, where 
    \begin{align}
    \label{OMD:gradient}
        g_{i}(\bm{p},r) = {\partial 
        \widetilde{\pi}_i(\bm{p}, r)}/{\partial p_{i}}= 2\beta_i p_{i} - \left(\alpha_{i} + \delta_{i} p_{-i} + \gamma_{i} r \right)\,.
    \end{align}
    This oracle can be viewed as a feedback mechanism that outputs the payoff gradient $\partial \widetilde{\pi}_{i}/\partial p_{i}$ evaluated at a given price profile $\bm{p}$ and reference price $r$.
    We note that the first-order feedback is very common in the optimization and learning literature as discussed in Section \ref{Sec:model}.
    Here, we point out that after a firm posts prices according to the OMD algorithm, it only obtains $g_{i,t}$, and does not necessarily observe the prices of its competitor nor the reference price.\footnote{Firms do not know the linear form of demand, and hence cannot learn parameters and then best respond given parameter estimates.} 
    \item Update proxy variable $y_{i,t+1}$ such that $R_{i}'(y_{i,t+1}) =  R_{i}'(p_{i,t})  - \epsilon_{i,t} g_{i,t},$\footnote{$y_{i,t+1}$ exists when $R_{i}$ is continuously differentiable and convex, see Section 3.3 of \cite{boyd2004convex} or Section 5.2 of \cite{bubeck2011introduction}} where we define $R'_{i}(q) \defeq \frac{dR_{i}(y)}{dy}\Big|_{y=q}$.
\end{enumerate}

We summarize the two-firm OMD pricing scheme in Algorithm \ref{algo:firmOMD}. 

\begin{minipage}{0.45\textwidth}
\begin{algorithm}[H]
    \centering
    \caption{2-firm OMD pricing under reference price updates}\label{algo:firmOMD}
    \footnotesize
    \begin{algorithmic}[1]
    \Require $\{R_{i},\{\epsilon_{i,t}\}_t\}_{i=1,2}$,  $y_{i,1} = \arg\min_{y\in \mathcal{P}}R_{i}(y)$ for $i = 1,2$.
        \For {$t= 1, 2,\ldots$}
        \For {$i=1,2$}
        \State Set price: $p_{i,t} = \Pi_{\mathcal{P}}(y_{i,t})$. 
        \State Access gradient $g_{i,t}= g_{i}(\bm{p}_{t},r_{t})$. 
        \State Update proxy variable: 
        $$R_i'(y_{i, t+1})= 
        R_i'(p_{i,t})
        - \epsilon_{i,t} g_{i,t}.$$
        \EndFor
        \State Reference price update (unobservable):$r_{t+1} = ar_{t}  + (1-a)\left(\theta_1 p_{1,t} + \theta_2 p_{2,t}\right)$
        \EndFor
    \end{algorithmic}
\end{algorithm}
\end{minipage}
\hfill
\begin{minipage}{0.48\textwidth}
\begin{algorithm}[H]
    \centering
    \caption{Induced 3-firm OMD pricing with no reference price}\label{algo:firmOMDinduced}
    \footnotesize
    \begin{algorithmic}[1]
      \Require $\{R_{i},\{\epsilon_{i,t}\}_t \}_{i=1,2,n}$, $y_{n,1} = r_{1}$, $y_{i,1} = \arg\min_{y\in \mathcal{P}}R_{i}(y)$ for $i=1,2$.
        \For {$t= 1, 2,\ldots$}
        \For {$i=1,2, n$}
        \State Set price: $p_{i,t} = \Pi_{\mathcal{P}}(y_{i,t})$. 
        \State Access gradient $g_{i,t}= g_{i}(\bm{p}_{t},r_{t})$. 
        \State Update proxy variable: 
        $$R_i'(y_{i,t+1}) =
        R_i'(p_{i,t})
        - \epsilon_{i,t} g_{i,t}.$$
        \EndFor
        \EndFor
        \vspace{1.5cm}
    \end{algorithmic}
\end{algorithm}
\end{minipage}
One can think of this sequential price competition with reference prices as a state-based dynamic game model where the reference price plays the role of an underlying state: each player (i.e., firm) has a continuous action space $\mathcal{P}$ and
payoff function $\widetilde{\pi}_{i}$ that depends on all players' actions as well as an underlying state variable $r_{t}$ that undergoes deterministic transitions.  However, the view that we will adopt in the rest of the paper perceives reference prices $r_{t}$ as price decisions $p_{n,t} = r_{t}$ posted by a virtual firm which we refer to as nature and denote it by $n$. This is possible if, for any $\widetilde{\pi}_{i}, R_{i}, \{\epsilon_{i,t}\}_{t}$ ($i=1,2$), we are able to construct a universal nature cost function $\widetilde{\pi}_{n}(p_{1},p_{2}, p_{n})$, strongly convex regularizer $R_{n}:\R \to \R$, and step size sequence $\{\epsilon_{n,t}\}_{t}$, such that when firms 1, 2 and nature independently run the OMD algorithm with their respective regularizers and step sizes (as summarized in Algorithm \ref{algo:firmOMDinduced}), the resulting price profiles $\{p_{1,t},p_{2,t},p_{n,t}\}_{t}$ recover the respective prices $\{\bm{p}_{t},r_{t}\}_{t}$ of Algorithm \ref{algo:firmOMD}. Here, note that $g_{n,t} = g_{n}(p_{1,t},p_{2,t},p_{n,t}) = {\partial \widetilde{\pi}_{n}(p_{1,t},p_{2,t},p_{n,t})}/{\partial p_{n,t}}$. The following Proposition \ref{lem:nature} formalizes this view and shows that such $\widetilde{\pi}_{n}, R_{n}$, and $\epsilon_{n,t}$ indeed exist. The proof is provided in Appendix \ref{app:OMD}, and we will refer to the dynamic game characterized in Algorithm \ref{algo:firmOMDinduced} as the \textit{induced 3-firm dynamic game}.
\begin{proposition}[Induced 3-firm dynamic  game]
\label{lem:nature}
Fix any $\widetilde{\pi}_{i}, R_{i}, \{\epsilon_{i,t}\}_{t}$, $i=1,2$, and initial reference price $r_{1}$. If nature (called firm $n$) is associated with cost function 
$\widetilde{\pi}_{n}(\bm{p},r) = \frac{1}{2} r^2 - \left(\theta_1 p_1 + \theta_2 p_2\right)r$, and 
chooses regularizer $R_{n}(r) = \frac{1}{2}r^{2}$ and step size $\epsilon_{n,t} = 1-a$, for any $t\ge 1$, then the price profiles $\{p_{1,t},p_{2,t},p_{n,t}\}_{t\geq 1}$  resulting from the game in Algorithm \ref{algo:firmOMDinduced} recovers the induced price and reference price trajectory $\{\bm{p}_{t},r_{t}\}_{t\geq 1}$ of Algorithm \ref{algo:firmOMD}.
\end{proposition}

We note that the  choices for nature's cost function $\widetilde \pi_n$, regularizer $R_n$ and step sizes $\{\epsilon_{n,t}\}_t$ may not be unique, and in Proposition \ref{lem:nature},  we simply choose the most straightforward feasible candidate.
 Nonetheless, by this lemma, the nature takes constant step sizes $1-a$, which implies that we have an inflexible (virtual) firm whose learning rate is always very fast.



By viewing reference prices as prices posted by nature, the induced 3-firm game is also associated with the static game that involves 3 players $i=1,2,n$ with respective  payoffs $\{ \widetilde{\pi}_{i}\}_{i=1,2,n}$ and common action set $\mathcal{P}$. It turns out that the \textit{pure strategy Nash Equilibrium (PSNE)} of this static game is unique and is identical to the SNE of Lemma \ref{lem:uniqueNE}:
\begin{proposition}[PSNE of induced 3-firm static game]
\label{lem:OMD:PSNE}
Consider the static game with players $i=1,2$ and nature $n$, who aims to minimize respective costs $\widetilde{\pi}_1, \widetilde{\pi}_2, \widetilde{\pi}_{n}$ with identical action set $\mathcal{P} = [\underline{p},\Bar{p}]$. Then, under Assumption \ref{assum:FOCinterior}, this game admits a unique PSNE $(\bm{p}^{*}, r^{*})$, i.e., $\widetilde{\pi}_{i}(p_{i}^{*}, \bm{p}_{-i}) \geq \widetilde{\pi}_{i}(p_{i}, \bm{p}_{-i}^{*})$ for $\forall p_{i} \in \mathcal{P}$ and $i = 1,2,n$. Furthermore, this PSNE is identical to the interior SNE of Lemma \ref{lem:uniqueNE}.
\end{proposition}

\section{Convergence Results}
\label{sec:firstOrder}
The key challenge in showing convergence for the induced 3-firm OMD game play in Algorithm  \ref{algo:firmOMDinduced} lies in the fact that the step size sequence for nature is the constant $1-a$, unlike previously studied multi-agent learning settings where step size sequences are typically identical across agents (see for example \cite{scutari2010convex, nagurney2012projected,bravo2018bandit,tampubolon2019pricing, mertikopoulos2019learning}). This highlights the fundamental issue in our problem of interest: \textit{will convergence still occur if one of the players takes a constant (fixed) step size?} 

In Section \ref{sec:decreasingsteps}, we show that prices and reference prices converge to the unique interior SNE  when the two firms adopt decreasing step sizes and characterize the corresponding convergence rate. 
In Section \ref{sec:constantsteps}, we show that there exist constant step sizes for the two firms with which prices  convergence to the SNE at faster rates compared to decreasing step sizes.  
\subsection{Decreasing Step Sizes}
\label{sec:decreasingsteps}
The first key result in this section is the following theorem, which states that if the two firms run the OMD algorithm with decreasing step sizes that do not go to zero too fast, then convergence to the SNE is guarantied.

\begin{theorem}[Convergence under Decreasing Step Sizes]\label{thm:OMD:convergenceSufDecStep}
Suppose that Assumption \ref{assum:FOCinterior} holds and firm $i=1,2$ adopts regularizer $R_{i}$ that is $\sigma_{i}$-strongly convex, constinuosly differentiable, and satisfies the reciprocity condition (see Section \ref{sec:OMD}). Then, when the sequence  $\{\epsilon_{i,t} =\epsilon_{t} \}_{t}$ is nonincreasing   with $\lim_{t \to \infty} \epsilon_{t} =  0$, we have  $\lim_{t\rightarrow \infty}(\sum_{i\in[2]}\theta_{i}p_{i,t} -r_t) \to 0$. Furthermore, if $\lim_{T \to \infty}\sum_{t=1}^{T}\epsilon_{t} = \infty$, $\lim_{T \to \infty}\sum_{t=1}^{T} \epsilon_{t}^{2} < \infty $ {and the sensitivity margin $\margin \geq 1$}, then $\left\{\bm{p}_{t}, r_{t}\right\}_{t}$ converges to the unique interior SNE $\left(\bm{p}^{*}, r^{*}\right)$. 
\end{theorem}
The first part of Theorem \ref{thm:OMD:convergenceSufDecStep} shows that prices  {stabilize} when the firms' step sizes go to zero eventually. This is an interesting result because in  the induced 3-firm dynamic game presented in Algorithm \ref{algo:firmOMDinduced}, nature adopts a constant step size and learns quickly, while the two other firms are learning slowly through decreasing step sizes.  {However, firms' prices may not necessarily converge, and even if they do,} firms may have the incentive to deviate, leading to an volatile  market.\footnote{ An example is the extreme case where firms 1 and 2 adopt step sizes $\epsilon_{i,t} = 0$. This obviously guaranties convergence because  prices are fixed at the initial prices, which are likely not the SNE, encouraging the firms to unilaterally deviate.} The second part of the theorem addresses this concern and shows that when $\lim_{T \to \infty}\sum_{t=1}^{T}\epsilon_{i,t} = \infty$ and $\lim_{T \to \infty}\sum_{t=1}^{T} \epsilon_{i,t}^{2} < \infty $, the market becomes stable as the prices converge to the SNE. {In fact, these conditions admit a large range of step sizes, e.g. $\epsilon_{i,t}= \Theta(1/t^{\eta})$ for $\eta \in (\frac{1}{2},1]$.} The proof is provided in Appendix \ref{app:firstOrder}. 
Here, we provide some examples to solidify  the aforementioned ideas. 
\begin{example}[Decreasing Step Sizes]
\label{ex:setting}
Consider the following demand and reference update model parameters: $\bm{\alpha} = (5,6)$, $\bm{\beta} = (2,3)$, $\bm{\delta} = (0.4,0.7)$, $\bm{\gamma} = (0.1,0.5)$, $\theta_{1} = 0.8$, $a = 0.4$, $\mathcal{P} = [1,2]$, and initial prices $(\bm{p}_{1},r_{1}) = (1,1,1.5)$. These parameters admit the unique SNE given by $(\bm{p}^{*},r^{*}) = (1.41, 1.28, 1.39)$. We consider two different decreasing step size sequences when both firms use the quadratic regularizer, i.e. $R_{1}(p) = R_{2}(p) = p^{2}/2$:
\begin{itemize}[leftmargin=*]
\vspace{-0.1cm}
\item 
With $\epsilon_{i,t} =0.1/t^{2}$, the price profile eventually converges to the point $(\widetilde{\bm{p}}, \widetilde{r}) = (1.21, 1.18, 1.20)$  which is not the SNE (see Figure \ref{fig:rates}a) and firms are incentivized to deviate, e.g., the best response for firm 1 w.r.t. $\widetilde{p}_{2}= 1.18$ and $\widetilde{r} = 1.20$ is $1.40 \neq \widetilde{p}_{1}$.\footnote{Here, we choose $\epsilon_{i,t} =0.1/t^{2}$ because the gap between the convergence point and the SNE is more visible. For the more natural choice $\epsilon_{i,t} =1/t^{2}$, we obtain similar results.
} Hence, under this step size sequence, firms may go through different epochs in the long run, in which firms converge in an epoch, and may decide to deviate and start over.
\vspace{-0.2cm}
\item  
With $\epsilon_{i,t} = 1/t$, we have  $\lim_{T \to \infty}\sum_{t=1}^{T}\epsilon_{i,t} = \infty$ and $\lim_{T \to \infty}\sum_{t=1}^{T} \epsilon_{i,t}^{2} < \infty $. Thus, per Theorem \ref{thm:OMD:convergenceSufDecStep},  prices and reference prices converge to the unique SNE; see Figure \ref{fig:rates}b. Moreover, we observe that (i) convergence occurs very quickly (for $t\ge 20$), and (ii) prices do not converge monotonically. The latter is in contrast with the  pricing policy presented in Theorem \ref{thm:NEconvergence}. 
\end{itemize}
\end{example}

In Example \ref{ex:setting}, we observe fast convergence to the SNE when firms choose decreasing step sizes. Inspired by this, we also characterize convergence rates for such step sizes:
\begin{theorem} [Convergence Rate under Decreasing Step Sizes]
\label{thm:OMD:rateDecStep}
Assume Assumption \ref{assum:FOCinterior} holds. For any sensitivity margin $\margin \geq 2$, if both firms adopt regularizer $R_{i}(\var) = \var^{2}$, there exists step sizes  $\epsilon_{i,t} = \Theta(1/t)$ and an absolute constant $\const$, {which depends on $a$ and $\max\{\theta_{1}, \theta_{2}\}$}, such that $\norm{\bm{p}^{*} - \bm{p}_{t}}^{2} \leq \const/t$ for any $t\in \N^{+}$. 
\end{theorem}
\vspace{-0.2cm}
The proof of this theorem constructs a sufficiently large absolute constant $\const$ and shows  $\norm{\bm{p}^{*} - \bm{p}_{t}}^{2} \leq \const/t$ via induction. The main procedure involves bounding $\norm{\bm{p}^{*} - \bm{p}_{t+1}}^{2}$ with $\norm{\bm{p}^{*} - \bm{p}_{t}}^{2}$ and $|r_{t}- r^{*}|$, and developing a tight bound for $\sum_{\tau = 1}^ {t-1} \norm{\bm{p}^{*} - \bm{p}_{\tau}}^{2}$. Bounding $\sum_{\tau = 1}^ {t-1} \norm{\bm{p}^{*} - \bm{p}_{\tau}}^{2}$ helps us bound $|r_t-r^*|$ because 
 the deviations of prices w.r.t. the interior SNE will cumulatively propagate into $|r_{t}- r^{*}|$ due to reference price update dynamics. The detailed proof is provided in Appendix \ref{app:firstOrder}. We also remark that the condition $\margin \geq 2$ is a rather practical regime because this condition, as discussed in Section \ref{Sec:model},  implies a firm's demand is more sensitive to its own prices compared to competitor's prices and surcharge (or discounts) relative to reference prices. {Finally, we remark that the constant $\const$ scales  reasonably  w.r.t. $a$ and $\max\{\theta_{1},\theta_{2}\}$ as long as they are bounded away from $1$; see Figure \ref{fig:rate_constant} in Appendix \ref{app:sec:addFigures} for an illustration for $a\in [0.1, 0.9]$ and $\max\{\theta_1, \theta_2\}\in \{0.5, 0.6, \ldots, 0.9\}$.}
\vspace{-0.4cm}
\begin{figure}[H]
		\centering
		\includegraphics[width=1.0\linewidth]{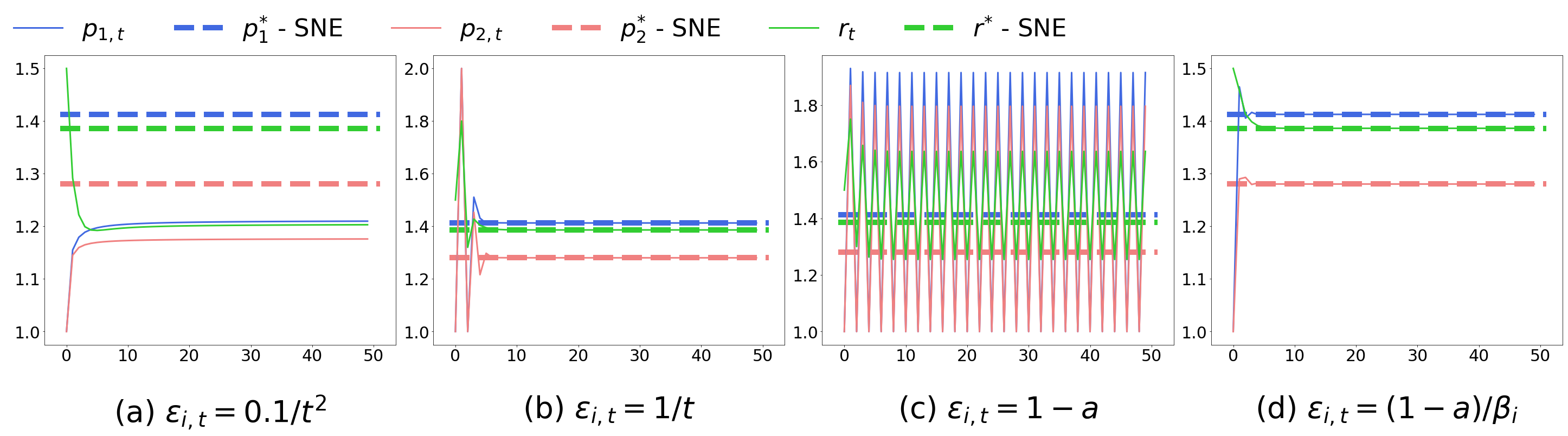}
	\caption{Illustration of price and reference price trajectories in Examples \ref{ex:setting} and \ref{ex:conststeps} under different step size sequences. The y-axis represents price levels, as the x-axis denotes time.}
		\label{fig:rates}
\end{figure}
\vspace{-0.7cm}
\subsection{Constant Step Sizes}
\label{sec:constantsteps}
\vspace{-0.1cm}
We start by revisiting Example \ref{ex:setting} and adopt constant step sizes. 
\begin{example}[Constant Step Sizes]
\label{ex:conststeps}
 Consider the same demand and reference update model parameters in Example \ref{ex:setting}. 
\begin{itemize}[leftmargin=*]
    \item  
    With $\epsilon_{i,t} = 1-a$, Figure \ref{fig:rates}c shows price profiles do not converge and oscillate in the long-run.
    \item 
    With $\epsilon_{i,t} = (1-a)/\beta_{i}$, Figure \ref{fig:rates}d shows price profiles converge to the SNE at a faster rate compared to decreasing step sizes in Figure \ref{fig:rates}b. 
\end{itemize}
\end{example}
\vspace{-0.2cm}
Given this example, we present the first main result for this section in the following theorem  (see proof in Appendix \ref{app:firstOrder}) which shows that under some conditions, 
there exists constant step size proportional to $\frac{1-a}{\beta_i}$ under which pricing profiles and reference price convergence to the unique interior SNE.

\begin{theorem}[Sufficient Conditions for Convergence under Constant Step Sizes]\label{thm:OMD:convergenceSuf}
Suppose that  firm $i$ adopts regularizer $R_{i}$ that is $\sigma_{i}$-strongly convex and continuously differentiable. For strong-convexity parameters $\sigma_{1},\sigma_{2}$ and sensitivity margin $\margin$, define the set $\mathcal{S}_{i,\margin}= \left\{\var > 0: f_{i,\margin}(\var)< 0 \right\}$, where 
\begin{align}
\begin{aligned}
\label{eq:OMD:constConvRegion}
    f_{i,\margin}(\var) = \begin{cases}
    \left(4\sigma_{i} +\frac{2\sigma_{-i}
    }{\margin^{2}}\right)\var^{2} - \left( \left(2- \frac{1}{2\margin}\right)\sigma_{i} - \frac{\sigma_{-i}}{2\margin}\right)\var + \frac{3}{4} & i=1,2\\
    \frac{2}{\margin^{2}}\left(\sigma_{1} + \sigma_{2}\right)\var^{2}+ \frac{1}{2\margin}  \left(\sigma_{1} + \sigma_{2}\right)\var -\frac{1}{4} & i = n
    \end{cases}\,.
\end{aligned}
\end{align}
Then, under Assumption \ref{assum:FOCinterior}, if $\cap_{i=1,2,n}\mathcal{S}_{i,\margin}\neq \emptyset$, the step size sequence  $\epsilon_{i,t} = \cone \sigma_{i} \frac{(1-a)}{\beta_{i}}$ ($i=1,2)$ for any $\cone\in \cap_{i=1,2,n}\mathcal{S}_{i,\margin}$ guarantees $\left\{\bm{p}_{t}, r_{t}\right\}_{t}$ converges to the unique interior SNE $\left(\bm{p}^{*}, r^{*}\right)$.
\end{theorem}
\vspace{-0.1cm}
This theorem indicates that under some conditions on $\margin,\sigma_{1}$, and $\sigma_{2}$, there exist constant step sizes 
with which  convergence 
to the unique interior SNE is guaranteed. The desired step size is proportional to $\frac{1-a}{\beta_i}$. This, roughly speaking, implies that prices converge to the SNE if firms adjust prices at a pace similar to that of nature. Recall that $1-a$ can be considered as the  step size of nature, and by demand model in Equation (\ref{def:RAdemand}), $\beta_i$ is firm $i$'s price sensitivity parameter. 
The conditions on $\margin,\sigma_{1}$, and $\sigma_{2}$ in Theorem \ref{thm:OMD:convergenceSuf} are, in fact, quite mild: the following Corollary \ref{cor:OMD:convergenceSuf} provides an example where for any $\margin >2$, we can find sufficiently large $\sigma_{1} = \sigma_{2}$ that guaranties convergence to an SNE.
\begin{corollary}[Convergence under Constant Step Sizes]
\label{cor:OMD:convergenceSuf}
 For any sensitivity margin $\margin > 2$, assume both firms adopt continuously differentiable regularizer $R_{i}$ that is $\sigma$-strongly convex where $\sigma > \sigma_0$ and $\sigma_0 \defeq \max\left\{\frac{6(2\margin^{2} + 1)}{(2\margin-1)^{2}}, \frac{\left(2\margin^{2} + 7\right)^{2}}{8\margin^{3} - 36\margin + 8} \right\}$. Then there exists constant $\cone$ dependent on $\margin$ and $\sigma$ so if firm $i\in \{1,2\}$ adopts step size  $\epsilon_{i,t} =\cone \sigma \frac{(1-a)}{\beta_{i}}$,  $\left\{\bm{p}_{t}, r_{t}\right\}_{t}$ converges to the unique interior SNE.
 \end{corollary}
This corollary provides sufficient conditions for the existence of constant step sizes that guarantee convergence to the SNE for any sensitivity margin $\margin > 2$. In fact, for suitable $\margin$, we can possibly find relatively small values of $\sigma_{1}$ and $\sigma_{2}$ such that the conditions are satisfied (e.g., $\sigma_{1} = \sigma_{2} = 4$ for $\margin = 5$). Note that $\sigma_{0} = \Theta(\margin)$ for large $\margin$, which means firms generally need to take larger strong-convexity parameters as $\margin$ increases. (See Figure \ref{fig:sig0K} in Appendix \ref{app:sec:addFigures} for illustration of $\sigma_{0}$ as a function of $\margin$.) Having everything else fixed, the larger $\sigma$, the slower price movements happens. \footnote{{For example, taking $R_{i}(\var) = \sigma \var^{2}$ in step 5 of Algorithm \ref{algo:firmOMD}, we get $y_{i,t+1} = p_{i,t} - \frac{\epsilon_{i,t}g_{i,t}}{\sigma}$, which implies the gap between $y_{i,t+1}$ and $p_{i,t}$ is small with large $\sigma$.}}
This is so because for large $\margin$, a firm's demand is very sensitive to its own prices, encouraging the firm to adjust prices slowly via large $\sigma$.

Moreover, we also characterize the convergence rate when firms adopt suitable constant step sizes via the following Theorem \ref{thm:OMD:constExpConvergence}, and highlight that such fast learning rates give us much faster convergence to the SNE,  compared to slow learning rates from decreasing step sizes.

\begin{theorem}[Convergence Rate for Constant Step Sizes]
\label{thm:OMD:constExpConvergence}
For any sensitivity margin $ \margin> 2$, assume that  both firms use quadratic regularizer $R_{i}(z) = \frac{\sigma z^{2}}{2}$ for any $\sigma > \sigma_0$, where $\sigma_0$ is defined in Corollary \ref{cor:OMD:convergenceSuf}. Then, under Assumption \ref{assum:FOCinterior} there exists constant $\cone > 0$, dependent on $\margin$ and $\sigma$, such that if firm $i=1,2$ adopts step size $\epsilon_{i,t} =\cone\sigma \frac{(1-a)}{\beta_{i}}$ for $t\in \N^{+}$, we have $\norm{\bm{p}^{*} - \bm{p}_{t}}^{2} \leq  \frac{1+2\sigma}{\sigma} \left(\Bar{p} - \underline{p}\right)^{2} \left(\frac{1+a}{2}\right)^{t}$. 
\end{theorem}

\subsection{Comparison with Multi-agent Online Learning}
\label{sec:comparison}

In light of Proposition \ref{lem:OMD:PSNE}, we can characterize the 3-player game consisting of firms and nature with
the mapping  $\bm{g}:\R_{+}^{3} \to \R_{+}^{3}$ s.t. $\bm{g}(\bm{p}) = (\partial \widetilde{\pi}_{i}/\partial p_{i})_{i=1,2,n}$, where we slightly abuse the notation and write $\bm{p} = (p_{1},p_{2},p_{n})$, and $\bm{g}(\bm{p}) = (g_{1}(\bm{p}),g_{2}(\bm{p}),g_{n}(\bm{p}))$. Note that the corresponding Jacobian of $\bm{g}$ is \[J = \begin{pmatrix}2\beta_{1},-\delta_{1},-\gamma_{1}\\ -\delta_{2}, 2\beta_{2},-\gamma_{2}\\-\theta_{1},-\theta_{2}, 1\end{pmatrix}\,,\] which is not necessarily positive definite,\footnote{Square matrix $A \in \R^{d\times d}$ is positive definite if for any $z\in \R^{d}$,  $z^{\top}Az > 0$. Note that $A$ does not need to be symmetric.}  despite being diagonally dominant\footnote{A square matrix $M = \{M_{ij}\}$ is diagonally dominant if $|M_{ii}| \geq \sum_{j\neq i} |M_{ij}|$.} due to our assumptions on model parameters as illustrated in Section \ref{Sec:model}. Also note that $\bm{g}(\bm{p}) = J\bm{p}$ is linear in $\bm{p}$, and Corollary 1.4 in \cite{nagurney2013network}  implies $\bm{g}$ is monotone if and only if $J$ is positive definite. Hence, in our setting, the mapping $\bm{g}$ may not be monotone, which prohibits us from naively applying arguments 
in the variational inequality (VI) framework to conclude convergence of the system as agents run OMD (see \cite{scutari2010convex,nagurney2012projected} for a detailed introduction on convergence to Nash Equilibrium under the VI framework).

Consequently, our proof techniques for Theorems 
\ref{thm:OMD:convergenceSufDecStep},\ref{thm:OMD:rateDecStep}, \ref{thm:OMD:convergenceSuf}, and \ref{thm:OMD:constExpConvergence} are not standard since the aforementioned mapping $\bm{g}$ does not necessarily satisfy monotonicity or other favorable properties that allow direct applications of the VI methodology. Even if we assume $\bm{g}$ is monotonic,  we still face technical issues that arise from heterogeneous step sizes, which provides another motivation to develop new techniques to show convergence as firms run general OMD algorithms.
To briefly illustrate such challenges, assume $\bm{g}$ is monotonic, meaning $\langle \bm{g}(\bm{p}), \bm{p}^{*} - \bm{p} \rangle \leq \langle \bm{g}(\bm{p}^{\star}), \bm{p}^{*} - \bm{p} \rangle = 0 $ for $\forall \bm{p}\in \mathcal{P}^{3}$, where the equality follows from Assumption \ref{assum:FOCinterior} and first order conditions. If one can enforce $\epsilon_{i,t} = \epsilon_{t}$ for $i=1, 2, n$, 
showing the convergence results in  this line of work boils down to verifying the following inequalities (e.g. see \cite{bravo2018bandit,tampubolon2019pricing, mertikopoulos2019learning}):
$$\sum_{i=1,2,n}\BDi(p_i^*, p_{i,t+1})~\overset{(a)}{\leq}~ \sum_{i=1,2,n} \BDi(p_i^*, p_{i,t}) +\epsilon_{t}\langle \bm{g}(\bm{p}_{t}), \bm{p}^{*} - \bm{p}_{t} \rangle +  \epsilon_{t}^{2}\V \overset{(b)}{<} \sum_{i=1,2,n} \BDi(p_i^*, p_{i,t}).$$
where $\V$ can be viewed as some absolute constant, and  $\BDi$ is Bregman divergence w.r.t. strongly convex regularizer $R_{i}$. At a high level, the above equations show that the distance between $p_i^*$ and $p_{i,t}$ becomes smaller over time  and hence implies convergence to the SNE. The inequality (a) follows from classical mirror descent proofs; and inequality (b) utilizes the variational stability condition by choosing suitable $\epsilon_{t}$ (for example $\epsilon_{t} = \Theta(1/t)$).  However, this procedure will not be applicable in our setting as nature is inflexible in the sense that it always takes the constant step size sequence $1-a$, while the two firms are unaware of how nature updates, and may independently use different step sizes (e.g. decreasing step sizes).

\bibliographystyle{informs2014}
\bibliography{reference}
\begin{center}
\vspace{0.8cm}
    \Large Appendices for\\
    \vspace{0.2cm}
    \Large \textbf{No-regret Learning in Price Competitions under Consumer Reference Effects}
    \noindent\makebox[\linewidth]{\rule{1\linewidth}{0.9pt}}
\end{center}

\section{Expanded Literature Review}
\label{app:litreview}
\textbf{Reference Price Effects and Monopolist Pricing.}  Consumer reference effects have been validated empirically in many works including \cite{tversky1979prospect,tversky1992advances,kalyanaram1995empirical,baron2019data}. This motivated a wide range of research including \cite{kopalle1996asymmetric,fibich2003explicit, popescu2007dynamic,ahn2007pricing, nasiry2011dynamic} that studies optimal dynamic monopolistic pricing under different demand and reference price update models, where the single firm has complete information on consumer demand as well as how reference prices update. 
 There are also very recent works that address the dynamic pricing problem with consumer reference effects under uncertain demand. \cite{baron2019data} utilizes real retail data and concludes the inclusion of exposure effects to sales or number of consumers\footnote{Exposure effects in reference price formation refer to considering reference prices as a weighted average of all historical prices, where weights depend on factors such as sales or number of consumers.} when considering reference price formations  leads to more accurate forecasts in demand, and proposes a pricing policy using dynamic programming. \cite{den2019dynamic} couples the problem of monopolistic dynamic pricing with reference effects and online demand learning. In our work, similar to  \cite{baron2019data,den2019dynamic}, firms do not know the demand functions and how the reference prices are formed. But, while in \cite{baron2019data,den2019dynamic},  the form of demand model is known to the firm (monopolist) that  aims to estimate model parameters, our work assumes competition between firms that do not know the form of demand and hence run OMD algorithms to increase revenue. Additionally, the algorithms proposed in \cite{baron2019data} and \cite{den2019dynamic} aim to increase revenue from the firm's perspective, while our work focuses on analyzing market stability for long-run competitions under reference effects.

\textbf{Pricing in Competitive Markets without Reference Effects.} A large stream of work studies static price competitions and characterizes structural properties of corresponding equilibria (for example, see \cite{bernstein2004general,gallego2006price,aksoy2013price}). Other works such as \cite{adida2010dynamic,levin2009dynamic,gallego2014dynamic} study oligopolistic dynamic pricing  under various inventory, market, or product characteristics. Nevertheless, these two lines of works are oblivious to consumer reference effects. In this work, we jointly tackle the dynamic pricing problems in competitive markets with reference price effects when the firms lack the knowledge of demand functions and reference price dynamics.

\textbf{Pricing in Competitive Markets with Reference Effects.} Similar to our work, the works of \cite{coulter2014pricing} and  \cite{federgruen2016price} also consider price competitions under reference effects.   \cite{coulter2014pricing} considers a similar linear demand model and an identical reference price update dynamic, but the work only provides theoretical analysis on the two-firm, two-period price competition setting, for which they characterize the unique sub-game perfect Nash Equilibrium. On the other hand, \cite{federgruen2016price} studies multiple-firm single-period price competition equipped with different reference price effects in consumers' demand (e.g. the reference price is specified by the lowest posted price). Additionally, both of these works study the complete information setting. In contrast to these two papers, our work studies price competitions over an infinite time horizon where reference prices adjust over time, and provides theoretical guarantees for the convergence of pricing strategies under the partial information setting. Finally, our work is the first study that provides theoretical analyses on long-term market stability of repeated price competitions in the presence of consumer reference effects.

\textbf{Convergence in Games with Descent Methods.} 
In addition to \cite{mertikopoulos2017convergence,bravo2018bandit,mertikopoulos2019learning} that we discussed in Section \ref{sec:intro}, here we also review related literature that study convergence in games where multiple agents adopt descent methods. \cite{rosen1965existence} studies finding a Nash Equilibrium of concave games via having each agent
run projected gradient descent under complete information, i.e., agents know each others' payoff functions and decision constraints. \cite{nedic2009distributed} studies a distributed network optimization problem to optimize a sum of convex objective functions corresponding to multiple agents. 
Our paper distinguishes itself from this line of work from two aspects: unlike the two aforementioned works, (i)   our model involves a varying underlying state (i.e., reference prices) dependent on all agents' historical decisions, and can be modeled as a sequence of decisions made by an inflexible virtual agent that adopts descent methods with a constant step size; 
(ii) the agents (i.e., firms) in our model do not have any information on one another's revenue function or how reference prices update. 
Finally, \cite{balseiro2019learning} considers multiple budget-constrained bidders participating in repeated second price auctions by adopting so-called \textit{adaptive pacing strategies}, which is equivalent to the subgradient descent method. 
In their setting, the subgradient for each bidder's objective is a function of all bidders' decisions as well as its budget rate (i.e. total fixed budget divided by a given time horizon), which can be thought of as an underlying model state that remains constant over time.\footnote{Note to run OMD algorithms in  \cite{balseiro2019learning}, agents need to know the length of the time horizon. Such knowledge is not required in our setting. } In contrast, in our setting, the  gradient oracle each firm receives is not only a function of all firms' decisions, but also of the reference price which varies over time according to firms' past decisions, making our analysis more challenging.


\section{Appendix for Section \ref{sec:NE}}
\label{app:NE}
\paragraph{Additional Definitions. } We define the best-response mapping as $\bm{\psi}:\mathcal{P}^{3} \to \mathcal{P}^{2}$ such that $\bm{\psi}(\bm{p},r) = \left(\psi_{1}(p_{2},r), \psi_{2}(p_{1},r) \right)$. Then, we can rewrite the set of best-response profiles w.r.t. reference price $r$, defined in Equation (\ref{eq:NE:bestresponseprfile}), as $\mathcal{B}(r) = \left\{\bm{p} \in \mathcal{P}^{2}: \bm{p} = \bm{\psi}(\bm{p},r)\right\}$. Note that for any SNE $(\bm{p}^{*},r^{*})$, we must have $\bm{p}^{*}\in \mathcal{B}(r^{*})$, and $\bm{p}^{*}$ is a fixed point of the mapping $\bm{\psi}(\cdot,r^{*})$.
\subsection{Proof of Theorem \ref{thm:NEconvergence}}

(i) By first order conditions, we know that \[\arg\max_{p\in \R} \pi_{i}(p,p_{-i},r) =\frac{\alpha_{i} + \delta_{i} p_{-i} + \gamma_{i} r }{2\beta_{i}}\,.\] 
Hence, due to boundary constraints on the decision set $\mathcal{P}$ and the revenue function being quadratic, we have \[\psi_{i}(p_{-i},r) =\arg\max_{p\in \mathcal{P}} \pi_{i}(p,p_{-i},r)= \Pi_{\mathcal{P}}\left(\frac{\alpha_{i} + \delta_{i} p_{-i} + \gamma_{i} r }{2\beta_{i}}\right)\,,\] where $\Pi_{\mathcal{P}}:\R \to \mathcal{P}$ is the projection operator such that $\Pi_{\mathcal{P}}(z) = z \I\{z \in \mathcal{P}\} + \underline{p} \I\{z < \underline{p}\}+\Bar{p} \I\{z > \Bar{p}\}$. Hence, $\psi_{i}(p_{-i},r)$ is a nondecreasing function in $p_{-i}$ and $r$, which further implies $\bm{\psi}(\bm{p},r)$ is nondecreasing in $\bm{p}$ and $r$. Again, recall for any $\bm{x}, \bm{y}$, the relationships $\bm{x} \leq \bm{y}$ and $\bm{y} \leq \bm{x}$ are component-wise comparisons.

We now follow a similar proof to that of Tarski's fixed point theorem: consider the set $\mathcal{B}_{+}(r) = \left\{\bm{p} \in \mathcal{P}^{2}: \bm{p}\leq \bm{\psi}(\bm{p},r) \right\}$. It is apparent that this set is nonempty because $(\underline{p},\underline{p}) \in \mathcal{B}_{+}(r)$. Fix any $\bm{p} \in \mathcal{B}_{+}(r)$. Then, we have $\bm{p}\leq \bm{\psi}(\bm{p},r)$ which further implies $\bm{\psi}(\bm{p},r)\leq \bm{\psi}\left(\bm{\psi}(\bm{p},r), r\right)$ since 
$\bm{\psi}(\bm{p},r)$ is nondecreasing in $\bm{p}$. Hence $\bm{\psi}(\bm{p},r) \in \mathcal{B}_{+}(r)$. By taking $\bm{U}(r) = \sup \mathcal{B}_{+}(r)$ (this is possible since all $\bm{p} \in \mathcal{B}_{+}(r)$ are bounded), we have $\bm{p} \leq \bm{U}(r)$ so  $\bm{p} \leq \bm{\psi}(\bm{p},r)\leq \bm{\psi}(\bm{U}(r),r)$. This further implies $\bm{U}(r)\leq \bm{\psi}(\bm{U}(r),r)$ because $\bm{U}(r)$ is the least upper bound of $\mathcal{B}_{+}(r)$, and thus $\bm{U}(r) \in\mathcal{B}_{+}(r) $. This allows us to conclude $ \bm{\psi}(\bm{U}(r),r) \leq \bm{U}(r)$ and hence $\bm{U}(r) = \bm{\psi}(\bm{U}(r),r)$, which means $\bm{U}(r) = \sup \mathcal{B}_{+}(r)$ is a fixed point of the mapping $\bm{\psi}(\cdot, r)$. Thus,  $\bm{U}(r)$ belongs in the set of best-response profiles $\mathcal{B}(r)$, confirming $\mathcal{B}(r)$ is not empty. 


Next, we  show that $ \mathcal{B}(r)$ is an ordered set with total ordering if it is not a singleton. To do so,  consider any $\bm{p}, \bm{q} \in \mathcal{B}(r)$ and without loss of generality assume $p_{1} > q_{1}$. Since $p_{1}  = \psi_{1}(p_{2},r)$ and $q_{1} =\psi_{1}(q_{2},r) $, by monotonicity of $\psi_{1}(\cdot,r)$ we have $p_{2} > q_{2}$. Thus, $\bm{p} >  \bm{q}$ and $ \mathcal{B}(r)$ is an ordered set with total ordering. 

(ii) In the proof of (i), we showed that $\bm{U}(r)= \sup \left\{\bm{p} \in \mathcal{P}^{2}: \bm{p}\leq \bm{\psi}(\bm{p},r) \right\}$ is a fixed point of the best-response mapping $\bm{\psi}(\cdot,r)$ for any $r$
which allows us to conclude $\bm{U}(r)$ is the largest best-response profile, i.e., $\bm{U}(r) = \max \mathcal{B}(r)$, and hence $\bm{p}_{t} = \bm{U}(r_{t})$. Furthermore, since $\bm{\psi}(\bm{p},r)$ is increasing in $r$, we know that $\bm{U}(\cdot)= \sup \left\{\bm{p} \in \mathcal{P}^{2}: \bm{p}\leq \bm{\psi}(\bm{p},\cdot) \right\}$ is also an increasing function. In the following, we will argue that the reference prices $r_{t}$ is monotonically increasing or decreasing, which implies  $\bm{p}_{t} = \bm{U}(r_{t})$ is also monotonic, and hence converges since prices and reference prices are bounded.

We write $\bm{U}(r) = \left(U_{1}(r), U_{2}(r)\right)$. At $t=1$, if $\theta_1 p_{1,1}  + \theta_2 p_{2,1} = \theta_1 U_{1}(r_{1})  + \theta_2 U_{2}(r_{1}) \geq r_{1}$, then the reference price at $t=2$ satisfies the following equation  \[r_{2} = ar_{1} + (1-a)\left( \theta_1 p_{1,1}  + \theta_2 p_{2,1}\right) \geq r_{1}\,.\] By the  monotonicity of $\bm{U}(\cdot)$, we have $p_{i,2} = U_{i}(r_{2}) \geq U_{i}(r_{1}) = p_{i,1}$ for $i = 1,2$. Thus, 
\begin{align*}
    r_{3} ~=~ &ar_{2} + (1-a)\left( \theta_1 p_{1,2}  + \theta_2 p_{2,2}\right) \\
~\geq~ & ar_{1} + (1-a)\left( \theta_1 p_{1,1}  + \theta_2 p_{2,1}\right) \\
~=~& r_{2}\,.
\end{align*}
A simple induction argument thus shows $\{r_{t}\}_t$ is a nondecreasing sequence. Since $r_{t} \leq \Bar{p}$ for any $t \in \N$, we know that $\{r_{t}\}_t$ converges to some number $r_{+} \in [\underline{p},\Bar{p}]$ when $\theta_1 p_{1,1}  + \theta_2 p_{2,1} \geq r_{1}$. Furthermore, we observe that $\lim_{t\to \infty}\bm{\psi}(\bm{U}(r_{t}), r_{t}) = \bm{\psi}(\bm{U}(r_{+}), r_{+})$  by the definition of $\bm{\psi}$. Also, from (i) we have $\bm{\psi}(\bm{U}(r_{t}), r_{t}) = \bm{U}(r_{t})$ and $\bm{\psi}(\bm{U}(r_{+}), r_{+}) = \bm{U}(r_{+})$ because $\bm{U}(r)$ is a fixed point of $\bm{\psi}(\cdot, r)$ for any $r$. Hence, $\lim_{t\to \infty}\bm{U}(r_{t}) =\bm{U}(r_{+}) $, which implies $\{\bm{p}_{t} = \bm{U}(r_{t})\}_t$ converges to $\bm{U}(r_{+})$. Note that convergence is monotonic because $\bm{U}(\cdot)$ is nondecreasing. Therefore, 
$$\theta_{1}U_{1}(r_{+}) + \theta_{2}U_{2}(r_{+}) = \lim_{t\to\infty} \theta_{1}U_{1}(r_{t}) + \theta_{2}U_{2}(r_{t}) =  \lim_{t\to\infty} r_{t+1} =  r_{+}\,,$$
which implies $\left(\bm{U}(r_{+}), r_{+} \right)$ is an SNE. We can thus conclude that if $\theta_1 p_{1,1}  + \theta_2 p_{2,1} = \theta_1 U_{1}(r_{1})  + \theta_2 U_{2}(r_{1}) \geq r_{1}$,
firms' prices and reference prices converge monotonically to an SNE $\left(\bm{U}(r_{+}), r_{+} \right)$.

Following a symmetric argument, if $\theta_1 p_{1,1}  + \theta_2 p_{2,1} < r_{1}$, we can show that 
$\{r_{t}\}_t$ is a nonincreasing sequence.  Since  
$r_{t} \geq \underline{p}$ for any $t \in \N$, we know that $\{r_{t}\}_t$ converges to some number $r_{-} \in [ \underline{p},\Bar{p}]$. Similar to the previous arguments, we can conclude  that prices and reference prices converge monotonically to an SNE $\left(\bm{U}(r_{-}), r_{-} \right)$.

\subsection{Proof of Lemma \ref{lem:uniqueNE}}
Let $(\bm{p}^{*}, r^{*}) \in (\underline{p},\Bar{p})^{3}$ be an interior SNE, whose existence is guarantied by Assumption \ref{assum:FOCinterior}. Since revenue functions are quadratic, first order conditions at the interior best-response profiles should hold, which means the derivative of revenue functions at the interior best-responses $p_{1}^{*} = \psi_{1}(p_{2}^{*}, r^{*})$ and $p_{2}^{*} = \psi_{2}(p_{1}^{*}, r^{*})$ should be 0:
$$\frac{\partial \pi_{1}(\bm{p}^{*}, r^{*})}{\partial p_1}= \frac{\partial \pi_{2}(\bm{p}^{*}, r^{*})}{\partial p_2} = 0\,,$$ 
which leads to the relationship $\alpha_{1} - 2\beta_{1} \psi_{1}(p_{2}^{*},r^{*}) + \delta_{1} p_{2}^{*} + \gamma_{1} r^{*} = \alpha_{2} - 2\beta_{2} \psi_{2}(p_{1}^{*},r^{*}) + \delta_{2} p_{1}^{*}+ \gamma_{2} r^{*} = 0$. Solving for the best-response equations, we get
\begin{align}
\label{eq:bestresponse}
 p_{1}^{*} = \psi_{1}(p_{2}^{*},r^{*}) = \frac{\alpha_{1} + \delta_{1} p_{2}^{*} + \gamma_{1} r^{*} }{2\beta_{1}} ~~~~,~~~~
p_{2}^{*} = \psi_{2}(p_{1}^{*},r^{*}) = \frac{\alpha_{2} + \delta_{2} p_{1}^{*} + \gamma_{2} r^{*} }{2\beta_{2}} \,.
\end{align}
Finally, the definition of an SNE guaranties  $r^{*} = \theta_{1} p_{1}^{*} + \theta_{2} p_{2}^{*} $. Thus, solving for $(\bm{p}^{*}, r^{*})$, we obtain the unique solution
 \begin{align}
 \begin{aligned}
 \label{eq:SNE}
     & p_{i}^{*} = \frac{ 2\alpha_{i} \beta_{-i} - \alpha_{i}\theta_{-i} \gamma_{-i} + \alpha_{-i}\left(\delta_{i} +\theta_{-i}\gamma_{i}\right)}{\left(2\beta_1 - \theta_1 \gamma_1\right)\left(2\beta_2 - \theta_2 \gamma_2\right) - \left(\theta_2\gamma_1 +\delta_1\right)\left(\theta_1\gamma_2 +\delta_2\right)} \quad i=1,2 \\
  & r^{*}= \frac{\theta_1\left(2\alpha_1\beta_2 + \alpha_2\delta_1\right)+\theta_2\left(2\alpha_2\beta_1 + \alpha_1\delta_2\right)}{\left(2\beta_1 - \theta_1 \gamma_1\right)\left(2\beta_2 - \theta_2 \gamma_2\right) - \left(\theta_2\gamma_1 +\delta_1\right)\left(\theta_1\gamma_2 +\delta_2\right)}\,.
  \end{aligned}
 \end{align}
 This implies that under Assumption \ref{assum:FOCinterior}, the interior SNE is unique. We remark that for any $i=1,2$, because $\beta_{i} \geq \margin(\delta_{i} + \gamma_{i})>0$ and $\margin \geq 2 > 1$ we have 
 $2\beta_{i} - \theta_{i} \gamma_{i} >\beta_{i} - \theta_{i} \gamma_{i}  > \delta_{i} + \gamma_{i}  - \theta_{i} \gamma_{i} = \theta_{-i} \gamma_{i} + \delta_{i}$. Hence, $p_{i}^{*} , r^{*} > 0$.

\section{Appendix for Section \ref{sec:OMD}}
\label{app:OMD}
\subsection{Proof of Proposition  \ref{lem:nature}}
First of all, it is easy to see prices at the first period are identical between Algorithm \ref{algo:firmOMD} and \ref{algo:firmOMDinduced}: $p_{i,1} = \arg\max_{p\in \mathcal{P}}R_{i}$ for $i=1,2$ and $p_{n,1} = r_{1}$. We now use induction to show price trajectories of the two algorithms are identical via considering the induction hypothesis that prices and reference prices are the same up to period $t \in \N^+$. 

 Note that $R_{n}(\var) = \frac{1}{2} z^{2}$ implies $R_{n}'(\var) = \var$. Then, the proxy variable update step for nature is 
\begin{align*}
    y_{n,t+1} ~=~ & p_{n,t} - (1-a) \frac{\partial \widetilde{\pi}_{n}(\bm{p})}{\partial p_{n}}\Big|_{\bm{p} = p_{1,t},p_{2,t},p_{n,t}} \\
    ~=~ & p_{n,t} - (1-a)\left(p_{n,t}  - \theta_1 p_{1,t} - \theta_2 p_{2,t}\right) \\
    ~=~ & ar_{t} + (1-a)\left(\theta_1 p_{1,t} + \theta_2 p_{2,t}\right)\\
    ~=~ &  r_{t+1}.
\end{align*}
Since $y_{n,t+1} = r_{t+1} \in \mathcal{P} $ the projection step for nature is trivial, which means $p_{n,t+1} = y_{n,t+1} = r_{t+1}$. Furthermore, it is not difficult to see that prices $p_{1,t+1} = \Pi_{p\in \mathcal{P}}(y_{1,t+1})$ and $p_{2,t+1} =\Pi_{p\in \mathcal{P}}(y_{2,t+1})$ are identical between the two algorithms under the induction hypothesis.  This implies that Algorithm \ref{algo:firmOMDinduced} indeed recovers the prices and reference prices produced by Algorithm \ref{algo:firmOMD}. 

\subsection{Proof of Proposition  \ref{lem:OMD:PSNE}}
Directly considering first order conditions for the  cost functions $\{ \widetilde{\pi}_{i}\}_{i=1,2,n}$, we have the system of equations
\begin{align*}
    & 0 = \frac{\partial \widetilde{\pi}_{1}(p_{1},p_{2},p_{n})}{\partial p_{1}} = 2\beta_{1}p_{1} - \left(\alpha_{1} + \delta_{1}p_{2} + \gamma_{1}r \right)\\
    & 0 = \frac{\partial \widetilde{\pi}_{2}(p_{1},p_{2},p_{n})}{\partial p_{2}} = 2\beta_{2}p_{2} - \left(\alpha_{2} + \delta_{2}p_{1} + \gamma_{2}r \right)\\
     &0 = \frac{\partial \widetilde{\pi}_{n}(p_{1},p_{2},p_{n})}{\partial p_{n}} =  p_{n} - \left( \theta_{1} p_{1} + \theta_{2} p_{2} \right)\,.
\end{align*}
Solving these equations results in a unique solution that is identical to that in Equation (\ref{eq:SNE}), which is the unique interior SNE according to Lemma \ref{lem:uniqueNE}. Since the SNE is an interior point of $(\underline{p},\Bar{p})^3$, it is the unique PSNE of the induced static 3-firm game.

\section{Appendix for Section \ref{sec:firstOrder}}
\label{app:firstOrder}

\subsection{Additional Definitions}
\begin{definition}[Bregman Divergence]
The Bregman divergence $ \BD: \mathcal{C} \times \mathcal{C} \to \R^+$ associated with convex set $\mathcal{C}\subset \R$, and convex and continuously differentiable function $R: \mathcal{C} \to \R$ is defined as \[  \BD(x,y) \defeq R(x) - R(y) - R'(y) (x-y) \geq 0\,,\] where the inequality follows from convexity of $R$. Furthermore, if $R$ is $\sigma$-strongly convex, then $ \BD(x,y) \geq \frac{\sigma^2}{2}(x-y)^{2}$. 
\end{definition}
Note that $\BDi$ is the Bregman divergence associated with regularizer $R_i$ used by firm $i=1,2$, and $\BDn$ is Bregman divergence associated with regularizer $R_n$ used by nature.

\begin{definition}
Let $g_i^{*}$ be the partial derivative of the cost function $\widetilde{\pi}_{i}$ w.r.t. $p_{i}$ evaluated at the interior SNE $(\bm{p}^{*},r^{*})$, i.e. for $i=1,2,n$
\begin{align*}
     g_i^* = \frac{\partial \widetilde{\pi}_{i}(p_{1},p_{2},p_{n})}{\partial p_i}\Big|_{p_{1} = p_{1}^{*}, p_{2} = p_{2}^{*}, p_{n} = r^{*}}\,.
\end{align*}
\end{definition}

\subsection{Proof for Theorem \ref{thm:OMD:convergenceSufDecStep}}

The proof of this theorem is divided into two parts. In the first part, we show that the price profiles $\left(\bm{p}_{t}, r_{t}\right)$ converge  as $t \to \infty$ under the condition $\lim_{t \to \infty} \epsilon_{i,t} =  0$, $i\in \{1,2\}$. In the second part, under the additional conditions $\lim_{T \to \infty}\sum_{t=1}^{T}\epsilon_{i,t} = \infty$ and $\lim_{T \to \infty}\sum_{t=1}^{T}\epsilon_{i,t}^{2} < \infty $, we show that the price profiles converge to the unique interior SNE.

\paragraph{First part: Convergence of prices and reference prices.} Recall that  $g_{i,t} = g_{i}(p_{i,t},r_{t}) = 2\beta_i p_{i,t} - \left(\alpha_i + \delta_i p_{-i,t} + \gamma_i r_{t} \right)$,  and $p_{i,t}$ and $p_{-i,t}  $ are both bounded. Hence, because for $i = 1,2$,  $\lim_{t \to \infty} \epsilon_{i,t} =  0$ and $\{ \epsilon_{i,t}\}_{t}$ is nonincreasing, we have for any small $\epsilon > 0$ there exist $ \te \in \N$ such that $|\epsilon_{i,t} g_{i,t}|\leq \frac{\sigma_{i}^{2}\epsilon}{6}$ for all $t \geq \te$. Our goal is to show that for $t\ge \te$, $|p_{i, t+1} - p_{i,t}|$ is small.

For $t \geq \te$,
\begin{align}
\label{eq:OMD:convergenceSufDecStep1}
    \left| p_{i,t+1} -p_{i,t} \right| ~\leq~ &  \left| p_{i,t+1} -y_{i,t+1} \right| + \left| y_{i,t+1} -p_{i,t} \right|
    ~\overset{(a)}{\leq}~  \left| p_{i,t+1} -y_{i,t+1} \right| + \frac{\epsilon}{6}\,. 
\end{align}
To see why inequality (a) holds recall that  $y_{i,t+1}$ is the proxy variable in Step 5 of Algorithm \ref{algo:firmOMD} such that
$R_{i}'(y_{i,t+1}) - R_{i}'(p_{i,t}) = \epsilon_{i,t} g_{i,t}$. Hence,
\begin{align*}
    \sigma_{i}^{2}\left|y_{i,t+1} -  p_{i,t}\right|{\leq}\left|R_{i}'(y_{i,t+1}) - R_{i}'(p_{i,t}) \right| =  \left|\epsilon_{i,t} g_{i,t}\right| \leq \frac{\sigma_{i}^{2}\epsilon}{6}, \quad t> \te, i = 1,2\,,
\end{align*}
which implies that $\left| y_{i,t+1} -p_{i,t} \right|\le \frac{\epsilon}{6}$, as desired. 
Here,  the first inequality holds because:
\begin{align*}  
\sigma_{i}^{2}\left( y_{i,t+1} -  p_{i,t}\right)^{2}& ~\overset{(a)}{\leq}~ \left(R_{i}'(y_{i,t+1}) - R_{i}'(p_{i,t})\right)\left( y_{i,t+1} -  p_{i,t}\right) \\
&~\leq~ \left|R_{i}'(y_{i,t+1}) - R_{i}'(p_{i,t})\right| \cdot \left|y_{i,t+1} -  p_{i,t} \right|\,,
\end{align*}
where (a) follows from summing up $R_{i}(y_{i,t+1}) - R_{i}(p_{i,t}) \geq R_{i}'(p_{i,t})(y_{i,t+1} - p_{i,t}) + \frac{\sigma_{i}^{2}}{2}\left( y_{i,t+1} -  p_{i,t}\right)^{2}$ and $R_{i}(p_{i,t}) - R_{i}(y_{i,t+1}) \geq R_{i}'(y_{i,t+1})(p_{i,t} - y_{i,t+1}) + \frac{\sigma_{i}^{2}}{2}\left( y_{i,t+1} -  p_{i,t}\right)^{2} $ due to strong convexity.

By Equation \eqref{eq:OMD:convergenceSufDecStep1}, for $t \geq \te$,
\begin{align}
    \left| p_{i,t+1} -p_{i,t} \right| 
     ~\le~ &  \left| p_{i,t+1} -y_{i,t+1} \right|\left(\I\{y_{i,t+1} < \underline{p}\}+ \I\{y_{i,t+1}\in \mathcal{P}\} + \I\{y_{i,t+1} > \Bar{p}\} \right) + \frac{\epsilon}{6} \nonumber \\
      ~{=}~ &  \left| p_{i,t+1} -y_{i,t+1} \right|\left(\I\{y_{i,t+1} < \underline{p}\}+ \I\{y_{i,t+1} > \Bar{p}\} \right) + \frac{\epsilon}{6}\,,
\end{align}
where the equality holds because under the event $y_{i,t+1}\in \mathcal{P}$, no projection occurs and hence, $y_{i,t+1} = p_{i,t+1}$. In the first of the proof, we bound the first two terms in the right hand side, i.e., $\left| p_{i,t+1} -y_{i,t+1} \right|\I\{y_{i,t+1} < \underline{p}\}$ and  $\left| p_{i,t+1} -y_{i,t+1} \right|\I\{y_{i,t+1} > \Bar{p}\}$.

To bound $\left| p_{i,t+1} -y_{i,t+1} \right|\I\{y_{i,t+1} < \underline{p}\}$, similar to Equation (\ref{eq:OMD:convergenceSufDecStep1}) we use $|p_{i,t}  - y_{i,t+1}| \leq \frac{\epsilon}{6}$ for $t \geq \te$ which implies $p_{i,t}  - y_{i,t+1} \leq \frac{\epsilon}{6}$. Thus,
\begin{align}
\label{eq:OMD:convergenceSufDecStep1.1}
 y_{i,t+1} \geq p_{i,t} - \frac{\epsilon}{6} \overset{(a)}{\geq} \underline{p} - \frac{\epsilon}{6}\,.
\end{align}
where (a) holds because $p_{i,t}  \geq \underline{p}$ for any $i,t$. On the other hand,  under the event $y_{i,t+1} < \underline{p}$, projection occurs and therefore we have $p_{i,t+1} = \underline{p}$. 

This yields
\begin{align}
    \left| p_{i,t+1} -y_{i,t+1} \right|\I\{y_{i,t+1} < \underline{p}\} = \left( \underline{p} - y_{i,t+1}\right)\I\{y_{i,t+1} < \underline{p}\} \overset{(a)}{\leq}\left( \underline{p} - \underline{p} + \frac{\epsilon}{6}\right) =  \frac{\epsilon}{6}\,,
\end{align}
where (a) follows from Equation (\ref{eq:OMD:convergenceSufDecStep1.1}).

Using a similar argument as above to bound $\left| p_{i,t+1} -y_{i,t+1} \right|\I\{y_{i,t+1}> \high\}$, we have $ \left| p_{i,t+1} -y_{i,t+1} \right|\I\{y_{i,t+1} >  \Bar{p}\}\leq \frac{\epsilon}{6}$ under the event $y_{i,t+1} > \Bar{p}$. 

Hence, plugging these upper bounds back into Equation (\ref{eq:OMD:convergenceSufDecStep1}), we can show that  for any $\epsilon > 0$ and $ t \geq \te$
\begin{align}
\label{eq:OMD:convergenceSufDecStep2}
    \left| p_{i,t+1} -p_{i,t} \right| ~\leq~ \frac{\epsilon}{6} + \frac{\epsilon}{6}+\frac{\epsilon}{6}~=~ \frac{\epsilon}{2}, \quad  i = 1,2 \,.
\end{align}
Now, for any  $t \geq \te$ we have
\begin{align}
    \left|r_{t+1} - \sum_{i=1,2}\theta_{i}p_{i,t+1}\right| ~=~ & \left|ar_{t} - (1-a)\sum_{i=1,2}\theta_{i}p_{i,t}- \sum_{i=1,2}\theta_{i}p_{i,t+1}\right|\\
    ~\leq~ & a\left|r_{t} -\sum_{i=1,2}\theta_{i}p_{i,t}\right| + \sum_{i=1,2}\theta_{i}\left|p_{i,t}- p_{i,t+1}\right|\\
     ~\leq~ & a\left|r_{t} -\sum_{i=1,2}\theta_{i}p_{i,t}\right| + \frac{\epsilon}{2}\,,
\end{align}
where the final inequality follows from Equation (\ref{eq:OMD:convergenceSufDecStep2}). Telescoping from $t$ down to $\te$, we have 
\begin{align*}
    \left|r_{t+1} - \sum_{i=1,2}\theta_{i}p_{i,t+1}\right| ~\leq~ &  a^{t-\te+1}\left|r_{\te} -\sum_{i=1,2}\theta_{i}p_{i,\te}\right|  + \frac{\epsilon}{2}\sum_{\tau = \te }^{t}a^{\tau-\te}\\
    ~\leq~ &  a^{t-\te+1}\left|r_{\te} -\sum_{i=1,2}\theta_{i}p_{i,\te}\right|  + \frac{\epsilon}{2(1-a)}\\
     ~\leq~ &  a^{t-\te+1}(\Bar{p}-\underline{p})  + \frac{\epsilon}{2(1-a)}
    \,.
\end{align*}
Letting $t \to \infty$ and $\epsilon \to 0$ concludes $\sum_{i=1,2}\theta_{i}p_{i,t} \to r_{t}$ for $t \to \infty$.

\paragraph{Second part: Convergence to the SNE.} The proof of this part is inspired by the proof of Theorem 4.6 in \cite{mertikopoulos2019learning}. However, in that proof, they rely on the Nash Equilibria of the game being \textit{variationally stable}, or more strictly speaking, the gradient of the virtual 3-player game (consisting of firms and nature) $\bm{g}:\R_{+}^{3} \to \R_{+}^{3}$ s.t. $\bm{g}(\bm{p}) = (\partial \widetilde{\pi}_{i}/\partial p_{i})_{i=1,2,n}$ to be a monotone mapping. In this proof, we do not rely on such structural assumption for the virtual 3-player game (see discussion in Section \ref{sec:comparison}).

The proof of this theorem is split into two steps. First, we show that for any $\epsilon > 0$, the price profile $(p_{1,t},p_{2,t})$ must enter an $\epsilon$-neigborhood of the SNE prices $(p_{1}^{*},p_{2}^{*})$ infinitely many times. In the second step, we show that when $(p_{1,t},p_{2,t})$ enters the $\epsilon$-neigborhood with small enough step sizes, it must remain their forever. 

Before we begin, we first introduce a lemma that would be used in both steps:

\begin{lemma}
\label{lem:psdGameMatrix}
Assume for any $i=1,2$, $\beta_{i} \geq \margin\left(\delta_{1} + \delta_{2} + \max\{\gamma_{1},\gamma_{2}\}\right)$ for $\margin\geq1$ as described in Section \ref{Sec:model}. Consider the 2-by-2 matrix $M$
\begin{align}
\begin{aligned}
   &  M = \begin{pmatrix} 2\beta_{1} -\theta_{1}\gamma_{1} & -(\delta_{1} + \theta_{2}\gamma_{1})\\
-(\delta_{2} + \theta_{1}\gamma_{2})& 2\beta_{2} -\theta_{2}\gamma_{2} 
\end{pmatrix} \\
& M + M^{\top} = \begin{pmatrix} 4\beta_{1} -2\theta_{1}\gamma_{1} & -(\delta_{1} + \delta_{2} + \theta_{2}\gamma_{1} + \theta_{1}\gamma_{2} )\\
-(\delta_{1} + \delta_{2} + \theta_{2}\gamma_{1} + \theta_{1}\gamma_{2} )& 4\beta_{2} -2\theta_{2}\gamma_{2} 
\end{pmatrix}\,,
\end{aligned}
\end{align}
Then, for any $\bm{x}\neq \bm{0} \in \R^{2}$, $\bm{x}^{\top}M\bm{x} > 0$. Furthermore, for any $\epsilon > 0$,  define an  $\epsilon$-neighborhood around $(p_{1},p_{2})\in\R^{2}: \mathcal{N}_{\epsilon}(p_{1},p_{2}) = \{\bm{x}\in \R^{2}: \sum_{i=1,2}\BDi(x_{i},p_{i}) < \epsilon\}$ where $\BD_{i}(x,y) =R(x) - R(y) - R'(y) (x-y)$ is the Bregman divergence w.r.t. $R_{i}$. Assuming the regularizers satisfy the reciprocity condition, i.e. whenever $x\to y$ for $x,y \in \R$ we have $\BD_{i}(x,y)\to 0$, there exists some absolute constant $C_{\epsilon}>0$ such that
\begin{align}
\label{eq:supportmonotone}
    \left(\bm{x}-\bm{p}\right)^{\top}M\left(\bm{x}-\bm{p}\right)   \geq C_{\epsilon} \quad, \forall \bm{x} \notin \mathcal{N}_{\epsilon}(p_{1},p_{2}) \,.
\end{align}
\end{lemma}
\begin{proof}
Since, $\beta_{i} \geq \margin\left(\delta_{1} + \delta_{2} + \max\{\gamma_{1},\gamma_{2}\}\right)$ for $\margin \geq 1$, the sum of the first row of $M + M^{\top} = 4\beta_{1} - 2\theta_{1}\gamma_{1} -(\delta_{1} + \delta_{2} + \theta_{2}\gamma_{1} + \theta_{1}\gamma_{2} )> 0 $. Similarly, the sum of the second row of $M + M^{\top}$ is also strictly greater than 0, and hence 
$M + M^{\top}$ is a symmetric strict diagonally dominant matrix, and hence positive definite.\footnote{A symmetric square matrix has real eigenvalues, and a diagonally dominant square matrix has eigenvalues whose real parts are positive. Hence $M + M^{\top}$ has positive eigenvalues, and is thus positive definite.} Therefore, for any $\bm{x}\in \R^{2}$ and  $\bm{x} \neq \bm{0}$, we have $ \bm{x}^{\top}\left(M + M^{\top} \right)\bm{x} > 0$. Since $\bm{x}^{\top}M\bm{x}=  \bm{x}^{\top}M^{\top}\bm{x}$, this implies our desired result $\bm{x}^{\top}M\bm{x} > 0$

We now show Equation (\ref{eq:supportmonotone}). The reciprocity condition implies there exists some $\eta_{\epsilon}> 0$ such that when $\norm{\bm{x} - \bm{p}} < \eta_{\epsilon}$, it holds that $\sum_{i=1,2}\BDi(x_{i},p_{i}) < \epsilon$.  Hence, for any $\bm{x} \notin \mathcal{N}_{\epsilon}(p_{1},p_{2})$, it must be the case that $\norm{\bm{x} - \bm{p}}\geq \eta_{\epsilon}$. So, for $\bm{x} \notin \mathcal{N}_{\epsilon}(p_{1},p_{2})$, we have
\begin{align*}
    \left(\bm{x}-\bm{p}\right)^{\top}M\left(\bm{x}-\bm{p}\right)  ~=~& \frac{1}{2} \left(\bm{x}-\bm{p}\right)^{\top}\left(M + M^{\top} \right)\left(\bm{x}-\bm{p}\right)\\
    ~\geq~& \frac{1}{2}\lambda_{\min}(M + M^{\top})\norm{\bm{x}-\bm{p}}^{2} \\
    ~\geq~& \frac{1}{2}\lambda_{\min}(M + M^{\top})\eta_{\epsilon}^{2} \defeq C_{\epsilon} ,
\end{align*}
where $\lambda_{\min}(M + M^{\top})$ is the minimum eigenvalue of the matrix $M + M^{\top}$ (which is positive due to the fact that $M + M^{\top}$ is positive definite).
\end{proof}
Now, returning to the proof for Theorem \ref{thm:OMD:convergenceSufDecStep}.

\textbf{Step 1: show $(p_{1,t},p_{2,t})$ must enter an $\epsilon$-neigborhood of the SNE prices $(p_{1}^{*},p_{2}^{*})$ infinitely many times.} We use a contradiction argument. Fix any $\epsilon >0$ and let $C_{\epsilon} > 0$ be the absolute constant defined in Lemma \ref{lem:psdGameMatrix}. Assume by contradiction that $(p_{1,t},p_{2,t})$ only visits $\mathcal{N}_{\epsilon}(p_{1}^{*},p_{2}^{*}) = \left\{\bm{p}\in \R^{2}: \norm{\bm{p}-(p_{1}^{*},p_{2}^{*})} <\epsilon\right\}$ finitely many times, i.e. there exists some $t_{\mathcal{N}}$ s.t. $\norm{(p_{1,t},p_{2,t})- (p_{1}^{*},p_{2}^{*})}\geq \epsilon$ for all $t \geq t_{\mathcal{N}}$. Further, in the first part of this theorem we showed that $\sum_{i = 1,2}\theta_{i}p_{i,t} \to r_{t}$, so without loss of generality (by taking $t_{\mathcal{N}}$ large enough), we can also assume for some small $\eta> 0$ such that $2\eta\gamma(\Bar{p}-\underline{p}) < C_{\epsilon}$ (where $\gamma = \max\{\gamma_{1},\gamma_{2}\}$), we have 
$\left|r_{t} - \sum_{i = 1,2}\theta_{i}p_{i,t}\right|\leq \eta$ for all $t \geq t_{\mathcal{N}}$.

We start by deducing a recurrence relationship between $\BDi(p_i^*, p_{i,t+1})$ and $\BDi(p_i^*, p_{i,t})$ as followed:
\begin{align}
\begin{aligned}
\label{eq:OMD:standardAnalysis0}
 & \BDi(p_i^*, p_{i,t+1})\\
     ~\overset{(a)}{\leq}~ & \BDi(p_i^*, p_{i,t})  - \epsilon_{i,t}\left(g_i^{*} - g_{i,t}\right)\left( p_i^{*} - p_{i,t} \right) +  \frac{\left(\epsilon_{i,t}\right)^2 g_{i,t}^{2}}{2\sigma_i}\\
     ~\leq~ & \BDi(p_i^*, p_{i,t})  - \epsilon_{i,t}\left(2\beta_{i}p_{i}^{*} - \delta_{i}p_{-i}^{*} + \gamma_{i}r^{*} -2\beta_{i}p_{i,t} - \delta_{i}p_{-i,t} + \gamma_{i}r_{t} \right)\left( p_i^{*} - p_{i,t} \right) +  \frac{\left(\epsilon_{i,t}\right)^2 g_{i,t}^{2}}{2\sigma_i}\\
      ~\overset{(b)}{\leq}~ & \BDi(p_i^*, p_{i,t}) \\
      & - \epsilon_{i,t}\left(2\beta_{i}p_{i}^{*} - \delta_{i}p_{-i}^{*} + \gamma_{i}\left(\sum_{j=1,2}\theta_{j}p_{j}^{*}\right) -2\beta_{i}p_{i,t} - \delta_{i}p_{-i,t} + \gamma_{i}\left(\sum_{j=1,2}\theta_{j}p_{j,t}\right) \right)\left( p_i^{*} - p_{i,t} \right)\\
     & + \epsilon_{i,t}\eta \gamma_{i}\left(\Bar{p}-\underline{p}\right) +  \frac{\left(\epsilon_{i,t}\right)^2 g_{i,t}^{2}}{2\sigma_i}\\
       ~=~ & \BDi(p_i^*, p_{i,t})  - \epsilon_{i,t}\left(\left(2\beta_{i} -\theta_{i}\gamma_{i}\right)\left(p_{i}^{*} - p_{i,t}\right) - \left(\delta_{i} + \theta_{-i}\gamma_{i}\right)\left(p_{-i}^{*} - p_{-i,t}\right)\right)\left(p_{i}^{*} - p_{i,t}\right)\\
      & \quad + \epsilon_{i,t}\eta \gamma_{i}\left(\Bar{p}-\underline{p}\right) +  \frac{\left(\epsilon_{i,t}\right)^2 g_{i,t}^{2}}{2\sigma_i}
     \,,
\end{aligned}
\end{align}
where in (a) we directly evoked Corollary \ref{cor:SNEBregmanUpdate}; and in (b) we
used the fact that $r^{*}=\sum_{j=1,2}\theta_{j}p_{j}^{*}$ according to the definition of the SNE, and $\left|r_{t} - \sum_{j=1,2}\theta_{j}p_{j,t}\right|\leq \eta$ for $t \geq  t_{\mathcal{N}}$. Summing the above over $i=1,2$ and recalling $\epsilon_{i,t}=\epsilon_{t}$ we have 
\begin{align}
\begin{aligned}
\label{eq:OMD:standardAnalysis1}
   & \sum_{i=1,2}\BDi(p_i^*, p_{i,t+1}) \\ ~\overset{(a)}{\leq}~ & \sum_{i=1,2}\BDi(p_i^*, p_{i,t})  - \epsilon_{t}\left(\bm{p}^{
    *} - \bm{p}_{t}\right)^{\top}M \left(\bm{p}^{
    *} - \bm{p}_{t}\right)+  2\epsilon_{t}\eta \gamma \left(\Bar{p}-\underline{p}\right) +  \frac{\left(\epsilon_{t}\right)^2 \V}{\sigma}\\
    ~\overset{(b)}{\leq}~ & \sum_{i=1,2}\BDi(p_i^*, p_{i,t})  - \epsilon_{t}C_{\epsilon}+  2\epsilon_{t}\eta \gamma \left(\Bar{p}-\underline{p}\right) +  \frac{\left(\epsilon_{t}\right)^2 \V}{\sigma}\\
     ~=~ & \sum_{i=1,2}\BDi(p_i^*, p_{i,t})  - \epsilon_{t}\left(C_{\epsilon}- 2\eta \gamma \left(\Bar{p}-\underline{p}\right)\right) +  \frac{\left(\epsilon_{t}\right)^2 \V}{\sigma}
\end{aligned}
\end{align}
where in (a) we take some finite $\V > \max_{i\in\{1,2\},t\in \N^{+}}g_{i,t}^{2}$ 
for all $i,t$ by recalling  Equation (\ref{OMD:gradient}) which states $ g_{i,t} = g_{i}(\bm{p}_{t},r_{t}) = 2\beta_i p_{i,t} - \left(\alpha_{i} + \delta_{i} p_{-i,t} + \gamma_{i} r_{t} \right)$, and that $p_{i,t},r_{t}$ are bounded within $[\underline{p},\overline{p}]$ for all $i,t$; in (b) we applied Lemma \ref{lem:psdGameMatrix} since  $\norm{(p_{1,t},p_{2,t})- (p_{1}^{*},p_{2}^{*})}\geq \epsilon$ for all $t \geq t_{\mathcal{N}}$. 

Telescoping Equation (\ref{eq:OMD:standardAnalysis1}) from some large time period $T+1$ down to $t_{\mathcal{N}}$ we get
\begin{align}
    \label{eq:OMD:telescope}
   0~\leq~&   \sum_{i=1,2}\BDi(p_{i}^*, p_{i,T+1}) \nonumber \\
    ~\leq~& \sum_{i=1,2}\BDi(p_i^*, p_{i,t_{\mathcal{N}}})  - \left(C_{\epsilon}- 2\eta \gamma \left(\Bar{p}-\underline{p}\right)\right) \sum_{t =t_{\mathcal{N}}}^{T}\epsilon_{t}+  \frac{\V}{\sigma}\sum_{t =t_{\mathcal{N}}}^{T}\left(\epsilon_{t}\right)^2\,.
\end{align}
Rearranging terms in Equation \eqref{eq:OMD:telescope} and dividing both sides by $\sum_{t= t_{\mathcal{N}}}^{T} \epsilon_{t}$ we get
the following:
\begin{align*}
   \frac{-\sum_{i=1,2}\BDi(p_i^*, p_{i,t_{\mathcal{N}}}) }{\sum_{t= t_{\mathcal{N}}}^{T} \epsilon_{t}}  ~\leq~  - \left(C_{\epsilon}- 2\eta \gamma \left(\Bar{p}-\underline{p}\right)\right)  +  \frac{\V}{\sigma}\cdot \frac{\sum_{t= t_{\mathcal{N}}}^{T} \epsilon_{t}^{2}}{\sum_{t= t_{\mathcal{N}}}^{T} \epsilon_{t}} \,.
\end{align*} 
Since $\sum_{i=1,2}\BDi(p_i^*, p_{i,t_{\mathcal{N}}}) $ and $t_{\mathcal{N}}$ are finite, the condition $\lim_{T \to \infty}\sum_{t=1}^{T} \epsilon_{t} = \infty $ implies that $\lim_{T \to \infty}\sum_{t=t_{\mathcal{N}}}^{T} \epsilon_{t} = \infty$, and the condition $\lim_{T \to \infty}\sum_{t=1}^{T} \epsilon_{t}^{2} < \infty $ implies that $\lim_{T \to \infty}\sum_{t=t_{\mathcal{N}}}^{T} \epsilon_{t}^{2} <\infty$. Hence,
 $\frac{\sum_{t= t_{\mathcal{N}}}^{T} \epsilon_{t}^{2}}{\sum_{t= t_{\mathcal{N}}}^{T} \epsilon_{t}}= 0$ as $T\to \infty$. Finally, because $C_{\epsilon}- 2\eta \gamma \left(\Bar{p}-\underline{p}\right) >0$ due to our definition of $\eta$, as $T\to \infty$, the above left hand side goes to zero and the  right hand side goes to $-\left(C_{\epsilon}- 2\eta \gamma \left(\Bar{p}-\underline{p}\right)\right) < 0$. This is a contradiction, implying that $(p_{1,t},p_{2,t})$ must enter an $\mathcal{N}_{\epsilon}(p_{1}^{*},p_{2}^{*})$ infinitely many times.
 
 \textbf{Step 2: show when $(p_{1,t},p_{2,t})$  enters an $\epsilon$-neigborhood of the SNE prices $(p_{1}^{*},p_{2}^{*})$ for large $t$ (small step size), it must stay in the neighborhood.} Here, we will show that for any $\epsilon > 0$, if $(p_{1,t},p_{2,t}) \in \mathcal{N}_{\epsilon}(p_{1}^{*},p_{2}^{*})$ for some large $t$, then $(p_{1,\tau},p_{2,\tau}) \in \mathcal{N}_{\epsilon}(p_{1}^{*},p_{2}^{*})$ for all $\tau \geq t$. In fact, we only need to show that  $(p_{1,t},p_{2,t}) \in \mathcal{N}_{\epsilon}(p_{1}^{*},p_{2}^{*})$
 implies $(p_{1,t+1},p_{2,t+1}) \in \mathcal{N}_{\epsilon}(p_{1}^{*},p_{2}^{*})$
for large $t$, and the rest follows from an induction argument.
 
Consider two scenarios, namely $(p_{1,t},p_{2,t}) \in \mathcal{N}_{\frac{\epsilon}{2}}(p_{1}^{*},p_{2}^{*})$ and $(p_{1,t},p_{2,t}) \in \mathcal{N}_{\epsilon}(p_{1}^{*},p_{2}^{*})/\mathcal{N}_{\frac{\epsilon}{2}}(p_{1}^{*},p_{2}^{*})$.

\textit{Scenario 1:} If $(p_{1,t},p_{2,t}) \in \mathcal{N}_{\frac{\epsilon}{2}}(p_{1}^{*},p_{2}^{*})$, consider some $\teta > 0$ such that when $t\geq \teta$, we have 
$\left|r_{t} - \sum_{i = 1,2}\theta_{i}p_{i,t}\right|\leq \eta$ for some small $\eta$ that satisfies $\frac{\epsilon}{2} > 2\eta \gamma \left(\Bar{p}-\underline{p}\right)$. Following the same deduction as in Equation (\ref{eq:OMD:standardAnalysis1}), we have 
\begin{align}
\begin{aligned}
 \label{eq:OMD:standardAnalysis3}
    \sum_{i=1,2}\BDi(p_i^*, p_{i,t+1}) ~\leq~ & \sum_{i=1,2}\BDi(p_i^*, p_{i,t})  - \epsilon_{t}\left(\bm{p}^{
    *} - \bm{p}_{t}\right)^{\top}M \left(\bm{p}^{
    *} - \bm{p}_{t}\right)+  2\epsilon_{t}\eta \gamma \left(\Bar{p}-\underline{p}\right) +  \frac{\left(\epsilon_{t}\right)^2 \V}{\sigma}\\
    ~\overset{(a)}{\leq}~ & \frac{\epsilon}{2} +  2\epsilon_{t}\eta \gamma \left(\Bar{p}-\underline{p}\right) +  \frac{\left(\epsilon_{t}\right)^2 \V}{\sigma}\\
     ~=~ & \frac{\epsilon}{2} +  \epsilon_{t}\left(2\eta \gamma \left(\Bar{p}-\underline{p}\right) +  \frac{\epsilon_{t} \V}{\sigma}\right)\\
     ~\overset{(b)}{\leq}~ &  \epsilon\,.
\end{aligned}
\end{align}
In (a) we used Lemma \ref{lem:psdGameMatrix} such that $\left(\bm{p}^{*} -\bm{p}_{t}\right)^{\top}M \left(\bm{p}^{*} - \bm{p}_{t}\right) \geq 0$, and the fact that $(p_{1,t},p_{2,t}) \in \mathcal{N}_{\frac{\epsilon}{2}}(p_{1}^{*},p_{2}^{*})$ so  by definition of an $\epsilon$-neighborhood (see Lemma \ref{lem:psdGameMatrix}) $ \sum_{i=1,2}\BDi(p_i^*, p_{i,t+1})\leq \frac{\epsilon}{2}$. In (b), we considered large $t$ such that $\epsilon_{t} < \min\left\{1, \frac{\sigma}{\V} \left(\frac{\epsilon}{2} - 2\eta \gamma \left(\Bar{p}-\underline{p}\right)\right)\right\}$ and used the definition of $\eta$ such that 
$\frac{\epsilon}{2} > 2\eta \gamma \left(\Bar{p}-\underline{p}\right)$. 

\textit{Scenario 2:} If $(p_{1,t},p_{2,t}) \in \mathcal{N}_{\epsilon}(p_{1}^{*},p_{2}^{*})/\mathcal{N}_{\frac{\epsilon}{2}}(p_{1}^{*},p_{2}^{*})$, let $C_{\frac{\epsilon}{2}}$ be defined as in Lemma \ref{lem:psdGameMatrix}. Consider some $\teta' > 0$ such that when $t\geq \teta'$, we have 
$\left|r_{t} - \sum_{i = 1,2}\theta_{i}p_{i,t}\right|\leq \eta$ for some small $\eta$  that satisfies $C_{\frac{\epsilon}{2}} > 2\eta \gamma \left(\Bar{p}-\underline{p}\right)$. Following the same deduction as in Equation (\ref{eq:OMD:standardAnalysis1}), we have 
 \begin{align*}
   \sum_{i=1,2}\BDi(p_i^*, p_{i,t+1}) ~\leq~ & \sum_{i=1,2}\BDi(p_i^*, p_{i,t})  - \epsilon_{t}\left(C_{\frac{\epsilon}{2}}- 2\eta \gamma \left(\Bar{p}-\underline{p}\right)\right) +  \frac{\left(\epsilon_{t}\right)^2 \V}{\sigma} \\
   ~\leq~ & \epsilon  - \epsilon_{t}\left(C_{\frac{\epsilon}{2}}- 2\eta \gamma \left(\Bar{p}-\underline{p}\right)\right) +  \frac{\left(\epsilon_{t}\right)^2 \V}{\sigma} \,,
\end{align*}
where the final inequality follows from 
$(p_{1,t},p_{2,t}) \in \mathcal{N}_{\epsilon}(p_{1}^{*},p_{2}^{*})$. Taking large $t$ such that $\epsilon_{t} < \frac{\sigma}{\V}\left(C_{\frac{\epsilon}{2}}- 2\eta \gamma \left(\Bar{p}-\underline{p}\right)\right)$ we get $\sum_{i=1,2}\BDi(p_i^*, p_{i,t+1})\leq \epsilon$.

Combining the above two scenarios, we showed that for any $t > 0$ such that $t\geq \max\{\teta,\teta'\}$ and $\epsilon_{t} \leq \min\left\{1,  \frac{\sigma}{\V} \left(\frac{\epsilon}{2} - 2\eta 
\gamma \left(\Bar{p}-\underline{p}\right)\right),\frac{\sigma}{\V}\left(C_{\frac{\epsilon}{2}}- 2\eta \gamma \left(\Bar{p}-\underline{p}\right)\right)\right\}$, 
$(p_{1,t},p_{2,t}) \in \mathcal{N}_{\epsilon}(p_{1}^{*},p_{2}^{*})$
 implies $(p_{1,t+1},p_{2,t+1}) \in \mathcal{N}_{\epsilon}(p_{1}^{*},p_{2}^{*})$ and hence $(p_{1,\tau},p_{2,\tau}) \in \mathcal{N}_{\epsilon}(p_{1}^{*},p_{2}^{*})$ for all large enough $\tau$ (by induction).

\subsection{Proof of Theorem \ref{thm:OMD:rateDecStep}}
Following the same deduction in Equation (\ref{eq:OMD:standardAnalysis0}) we have
\begin{align*}
    \BDi(p_{i}^*, p_{i,t+1}) 
     ~\leq~ & \BDi(p_{i}^*, p_{i,t}) -  \epsilon_{i,t} \left( 2\beta_{i}\left(p_{i}^* - p_{i,t} \right)^2 - 
    \delta_{i} \left(p_{-i}^{*} -  p_{-i,t}\right) \left(p_{i}^* - p_{i,t} \right) \right.\nonumber\\
     &~~~~\left.  - \gamma_{i} \left(r^{*} - r_{t} \right) \left(p_{i}^* - p_{i,t} \right)
    \right)  + \frac{\left(\epsilon_{i,t}g_{i,t}\right)^2}{2\sigma_{i}} \nonumber \\ 
    ~\overset{(a)}{\leq}~ &  \BDi(p_{i}^*, p_{i,t}) -  \epsilon_{i,t} \left(\frac{4\beta_{i} - \delta_{i}}{2} \left(p_{i}^* - p_{i,t} \right)^2 - 
    \frac{\delta_{i}}{2} \left(p_{-i}^{*} -  p_{-i,t}\right)^2 \right) \nonumber \\ 
    ~~~~ & + \epsilon_{i,t} \gamma_{i} \left(r^{*} - r_{t} \right) \left(p_{i}^* - p_{i,t} \right)  + \frac{\left(\epsilon_{i,t}g_{i,t}\right)^2}{2\sigma_{i}} \,.
\end{align*}

where in (a) we used the basic inequality $AB \leq (A^{2} + B^{2})/2$ for $A= p_{-i}^{*} -  p_{-i,t}$ and $B = p_{i}^* - p_{i,t}$.

Now, consider the step-size sequences $\{\epsilon_{i,t}\}_{t}$ that satisfy
\begin{align}
\label{eq:OMD:stepsizeRange}
    \frac{1}{t+1}\cdot \frac{10}{4\beta_{i} - \delta_{i} }~\leq~ \epsilon_{i,t} ~\leq~ \frac{1}{t+1}\cdot \frac{2}{\max\{\delta_{i},\gamma_{i}\}},\quad i = 1,2\,.
\end{align}
Equation (\ref{eq:OMD:stepsizeRange}) holds due to the fact that $\beta_{i} > \margin( \delta_{i} + \gamma_{i}) > 0$ and $\margin\geq 2$, which further implies $2\left( 4\beta_{i} - \delta_{i} \right) > 8\margin(\delta_{i} + \gamma_{i}) - 2\delta_{i}  > 10(\delta_{i} + \gamma_{i}) > 10\max\{\delta_{i},\gamma_{i}\}$. This leads to
\begin{align*}
    &\BDi(p_{i}^*, p_{i,t+1}) \\
    ~\leq~ &  \BDi(p_{i}^*, p_{i,t}) -  \frac{5}{t+1} \left(p_{i}^* - p_{i,t} \right)^2 + \frac{1}{t+1} \left(p_{-i}^{*} -  p_{-i,t}\right)^2  + \frac{2}{t+1} \left(r^{*} - r_{t} \right) \left(p_{i}^* - p_{i,t} \right)    + \frac{\left(\epsilon_{i,t}g_{i,t}\right)^2}{2\sigma_{i}} \\
    ~\overset{(a)}{\leq}~ &  \BDi(p_{i}^*, p_{i,t}) -  \frac{5}{t+1} \left(p_{i}^* - p_{i,t} \right)^2 + \frac{1}{t+1} \left(p_{-i}^{*} -  p_{-i,t}\right)^2   + \frac{2}{t+1} \left(r^{*} - r_{t} \right) \left(p_{i}^* - p_{i,t} \right)  +  \frac{\V}{2(t+1)^{2}}
    \,, 
\end{align*}
where in (a) we take some $\V > \max_{i\in\{1,2\},t\in \N^{+}}\frac{4g_{i,t}^{2}}{\sigma_{i}\max\{\delta_{i},\gamma_{i}\}^{2}}$  
for all $i,t$ by using the fact that $p_{i,t},r_{t} \in \mathcal{P}$. Summing across $i = 1,2$, we have
\begin{align*}
    & \sum_{i=1,2}\BDi(p_{i}^*, p_{i,t+1}) \\
    ~\leq~ &   \sum_{i=1,2}\BDi(p_{i}^*, p_{i,t}) -  \frac{4}{t+1} \norm{\bm{p}^{*} - \bm{p}_{t}}^{2}  + \frac{2}{t+1} \left(r^{*} - r_{t} \right) \left(p_{1}^{*} - p_{1,t}+ p_{2}^{*} - p_{2,t} \right)  +  \frac{\V}{(t+1)^{2}}\\
     ~\overset{(a)}{\leq}~ &   \sum_{i=1,2}\BDi(p_{i}^*, p_{i,t}) -  \frac{2}{t+1} \norm{\bm{p}^{*} - \bm{p}_{t}}^{2}  + \frac{1}{t+1} \left(r^{*} - r_{t} \right)^{2} +  \frac{\V}{(t+1)^{2}}
    \,, 
\end{align*}
where in inequality (a) we applied  $C(A+B) \leq \frac{C^{2}}{2} + \frac{(A+B)^{2}}{2} \leq \frac{C^{2}}{2} + A^{2} + B^{2}$ for $A = p_{1}^{*} - p_{1,t}$, $B = p_{2}^{*} - p_{2,t}$ and $C =r^{*} - r_{t} $.

When $R_{i}(z) = z^{2}$, we have $\BDi(p,p') = (p-p')^{2}$. Therefore, denoting $x_{t} =   \sum_{i=1,2}\BDi(p_{i}^{*}, p_{i,t}) = \norm{\bm{p}^{*} - \bm{p}_{t}}^{2}$ for $i = 1,2$ and $x_{n,t} = \left(r^{*} - r_{t} \right)^{2}$, the equation above yields
\begin{align}
\label{eq:OMD:rateIteration0}
     x_{t+1} ~\leq~  \left(1-\frac{2}{t+1}\right)x_{t} +\frac{1}{t+1}x_{n,t} + \frac{\V}{(t+1)^{2}}\,.
\end{align}

We will show via induction that $x_{t}\leq \frac{\const}{t}$ for some $\const > 0$ and any $t\in \N^{+}$. The proof is constructive and will rely on the following definitions, whose motivations will later be clear.

Fix  $\NN = \left \lceil\frac{a}{1-a}\right \rceil + 1$, $\MN = \left\lceil \frac{\frac{a}{1-a}(\NN+1)}{\NN - \frac{a}{1-a}}\right\rceil$,  and take any $\ourrho$ such that $\max\left\{\theta_{1}, \theta_{2}\right\} < \ourrho < 1$. Here,  $\lceil x \rceil = \min \{y \in \N^{+}: y \geq x \}$ for any $x \in \R$. Note that $\NN$ is bounded as $a$ is bounded away from 1.

Next, define
\begin{align}
    & \trho \defeq \min\left\{\tau \in \N^{+}: \tau \geq \NN \text{ and } \frac{(\NN+1)\log(\tau-\NN-1)}{\tau} \leq  \frac{\ourrho}{\max\{\theta_{1}, \theta_{2} \}}-1 \right\} \label{def:OMD:Trho}\\
    & \Q \defeq (1-a)\max\{\theta_{1}, \theta_{2}\}\sum_{\tau=1}^{\NN+\MN-1} \frac{a^{-\tau}}{\tau}\,. \label{eq:OMD:Q}
\end{align}
Note that $\trho$ is bounded because $\max\{\theta_1, \theta_2\}$ is bounded away from one. Further, since $\NN$ and $\ourrho$ are constant, and $\log(t) =o(t)$, it is easy to see that $\trho$ exists. Furthermore, define
\begin{align*}
    & \tildet \defeq \min \Bigg\{\tau >  \max\left\{\NN + \MN, \trho\right\}:\\
    & ~~~~~~~~ \left. (t-\NN) \cdot \left(2(\Bar{p} -\underline{p})^{2} + \Q\cdot \frac{ 2t \cdot (\Bar{p} -\underline{p})^{2} + \V + 1}{1-\ourrho}\right) < a^{-t}\quad, \text{for } \forall t \geq \tau \right\}\,,\\
    &  \const \defeq\frac{ 2 \tildet (\Bar{p} -\underline{p})^{2} + \V + 1}{1-\ourrho}\,.
\end{align*}
Note that $\tildet$ must exist because  the left hand side is quadratic in $t$, while the right hand side is exponential in $t$ for $a \in (0,1)$. {We provide an illustration for the size of $\const$ w.r.t. memory parameter $a$ and $\max\{\theta_{1},\theta_{2}\}$ in Figure \ref{fig:rate_constant} of Appendix \ref{app:sec:addFigures}.}

Note that the definition of $\tildet$ and $\const$ implies that the following three equations hold 
\begin{align}
   & a^{t}\left(2(\Bar{p} -\underline{p})^{2} + \Q\const\right) < \frac{1}{t-\NN} \quad \forall t \geq \tildet \label{eq:OMD:tildeT1}\\
   &   2 \tildet (\Bar{p} -\underline{p})^{2} + \V + 1 + \ourrho \const = \const \label{eq:OMD:tildeT2}\\
   &  \const > 2\tildet (\Bar{p} -\underline{p})^{2} \label{eq:OMD:tildeT3}
   \,.
\end{align}
Here, Equation (\ref{eq:OMD:tildeT1}) is due to the following: plugging the definition of $\const$ into that of $\tildet$ we get
$(\tildet-\NN) \cdot \left(2(\Bar{p} -\underline{p})^{2} + \Q\const \right) < a^{-\tildet}$, and since $\const = \frac{ 2 \tildet (\Bar{p} -\underline{p})^{2} + \V + 1}{1-\ourrho} \leq   \frac{ 2t \cdot (\Bar{p} -\underline{p})^{2} + \V + 1}{1-\ourrho}$ for any $t \geq \tildet > \NN$, we have  
$$(t-\NN)\left(2(\Bar{p} -\underline{p})^{2} + \Q\const\right) \leq  (t-\NN)\left(2(\Bar{p} -\underline{p})^{2} + \Q\cdot \frac{ 2t \cdot (\Bar{p} -\underline{p})^{2} + \V + 1}{1-\ourrho}\right) \overset{(a)}{<} a^{-t}\,,$$
where (a) follows from the definition of $\tildet$. Equation (\ref{eq:OMD:tildeT2}) directly follows from the definition of $\const$. Equation (\ref{eq:OMD:tildeT3}) follows because $\ourrho \in (0,1)$ and hence $\const = \frac{ 2 \tildet (\Bar{p} -\underline{p})^{2} + \V + 1}{1-\ourrho} > 2 \tildet (\Bar{p} -\underline{p})^{2} + \V + 1 >  2 \tildet (\Bar{p} -\underline{p})^{2}$. Hence, this implies that  $x_{t} \leq 2(\Bar{p} -\underline{p})^{2}< \frac{\const}{t}$ for any $t = 1\dots \tildet$, where we recall $x_{t} = \norm{\bm{p}^{*} - \bm{p}_{t}}^{2}$.\\

Consider $t\ge \tildet$.  We will now show via induction that $x_{t+1} \leq \const/(t+1)$ using our induction hypothesis  that 
$x_{\tau}\leq \const/\tau$ holds for all $\tau =1, \dots, t$. 
Note that the base case $x_{t} \leq 2(\Bar{p} -\underline{p})^{2}< \frac{\const}{t}$ for any $t = 1\dots \tildet$ is trivially true as we just discussed. Then, multiplying $t(t+1)$ on both sides of the recurrence relation in Equation (\ref{eq:OMD:rateIteration0}) and telescoping from $\tildet$  to $t$, we have
\begin{align}
\label{eq:OMD:rateIteration}
     t(t+1)x_{t+1} ~\leq~&  (t-1)t x_{t} +t x_{n,t} + \V \nonumber \\
     ~\leq~&  (t-2)(t-1) x_{t-1} + \sum_{\tau=t-1}^{t} t x_{n,\tau} + 2\V \nonumber \\
      ~\vdots~& \nonumber \\
     ~\leq~& (\tildet-1) \tildet\cdot x_{\tildet} + \sum_{\tau=\tildet}^{t} \tau x_{n,\tau} + (t - \tildet+1)\V  \nonumber\\
     ~\leq~&  (\tildet-1) \tildet \cdot x_{\tildet} + \sum_{\tau= \tildet}^{t} \tau x_{n,\tau} + t\V  \,.
\end{align}
We will now bound $x_{n,\tau}$ for all $\tau = \tildet\dots t$. Using the definition $r^{*} = \theta_{1}p_{1}^{*} + \theta_{2}p_{2}^{*}$, we get
\begin{align*}
    r^{*} - r_{\tau+1} ~=~ &  r^{*} - ar_{\tau}-(1-a)\left(\theta_{1}p_{1,\tau} + \theta_{2}p_{2,\tau} \right) \nonumber \\
     ~{=}~ &  a\left( r^{*} - r_{\tau}\right) -(1-a)\left(\theta_{1}\left(p_{1}^{*} - p_{1,\tau}\right) + \theta_{2}\left(p_{2}^{*} - p_{2,\tau}\right)\right)
      \,.
\end{align*}
By convexity, we further have for any $\tau = 1\dots t$,
\begin{align*}
x_{n,\tau+1} = \left(r^{*} - r_{\tau+1}\right)^{2}  ~\leq~ &  a x_{n,\tau} + (1-a)\left(\theta_{1}\left(p_{1}^{*} - p_{1,\tau}\right) + \theta_{2}\left(p_{2}^{*} - p_{2,\tau}\right)\right)^{2}\\
 ~\leq~ & a x_{n,\tau} + (1-a)\left(\theta_{1}\left(p_{1}^{*} - p_{1,\tau}\right)^{2} + \theta_{2}\left(p_{2}^{*} - p_{2,\tau}\right)^{2}\right)\\
 ~\leq~ & a x_{n,\tau} + (1-a)\max\{\theta_{1}, \theta_{2}\}x_{\tau}\\
 ~\overset{(a)}{\leq}~ & a x_{n,\tau} + (1-a)\max\{\theta_{1}, \theta_{2}\}\frac{\const}{\tau}\,,
\end{align*}
where (a) follows from the induction hypothesis, i.e., $x_{\tau}\leq \const/\tau$ holds for all $\tau =1 \dots t$. Using a telescoping argument, we then have for any $t \geq \tildet$,
\begin{align*}
    x_{n, t+1}~\leq~ &a x_{n,t} + (1-a)\max\{\theta_{1}, \theta_{2}\}\frac{\const}{t} \\
    ~\leq~ & a^2x_{n, t-1}+ (1-a)\const\max\{\theta_1, \theta_2\}\sum_{\tau=t-1}^{t}\frac{a^{t-\tau}}{\tau}\\
    & \vdots \\
     ~\leq~ & a^{t}x_{n,1} +(1-a) \const\max\{\theta_{1}, \theta_{2}\} \sum_{\tau=1}^{t} \frac{a^{t-\tau}}{\tau}\\
     ~=~ & a^{t}x_{n,1}  +(1-a)a^{t} \const\max\{\theta_{1}, \theta_{2}\}\left( \sum_{\tau=1}^{\NN+\MN-1} \frac{a^{-\tau}}{\tau} + \sum_{\tau=\NN+\MN}^{t} \frac{a^{-\tau}}{\tau} \right)\\
      ~\overset{(a)}{=}~ & a^{t}\left( x_{n,1} + \Q \const\right) +  (1-a)a^{t} \const\max\{\theta_{1}, \theta_{2}\}
      \sum_{\tau=\NN+\MN}^{t} \frac{a^{-\tau}}{\tau}\\\
       ~\leq~ & a^{t}\left(2(\Bar{p}-\underline{p})^{2}+ \Q \const\right) +  (1-a)a^{t} \const\max\{\theta_{1}, \theta_{2}\}
      \sum_{\tau=\NN+\MN}^{t} \frac{a^{-\tau}}{\tau}\\
       ~\overset{(b)}{\leq}~ & \frac{1}{t-\NN} +  (1-a)a^{t} \const\max\{\theta_{1}, \theta_{2}\}
      \sum_{\tau=\NN+\MN}^{t} \frac{a^{-\tau}}{\tau}\\
     ~\overset{(c)}{\leq}~ &  \frac{1 + \max\{\theta_{1}, \theta_{2}\} \const}{t-\NN}\,.
\end{align*}
Here, (a) follows from the definition of $\Q$ in Equation (\ref{eq:OMD:Q}); (b) follows from Equation (\ref{eq:OMD:tildeT1}); and (c) follows from Lemma \ref{lem:boundSum} since $t \geq \tildet \geq \NN + \MN$.
Applying this upper bound on $x_{n,t}$ in Equation (\ref{eq:OMD:rateIteration}) we have
\begin{align*}
     & t(t+1)x_{t+1} \nonumber \\
     ~\leq~& (\tildet-1) \tildet \cdot x_{\tildet} + \left(1 + \max\{\theta_{1}, \theta_{2}\} \const\right) \sum_{\tau=\tildet}^{t} \frac{\tau}{\tau-\NN-1} + t\V \nonumber \\
     ~=~& (\tildet-1) \tildet\cdot x_{\tildet} + \left(1 + \max\{\theta_{1}, \theta_{2}\} \const\right) \sum_{\tau = \tildet}^{t} \left(1 + \frac{\NN+1}{\tau-\NN-1}\right) + t\V \nonumber\\
     ~\overset{(a)}{\leq}~& (\tildet-1) \tildet\cdot x_{\tildet} + \left(1 + \max\{\theta_{1}, \theta_{2}\} \const\right) \left(t + (\NN+1)\log(t-\NN-1)\right) +t\V
     \,.
\end{align*}
 where (a) follows from $\sum_{\tau = \tildet}^{t}\frac{1}{\tau-\NN-1} \leq \int_{\tildet-1}^{t}\frac{1}{\tau-\NN-1}  d\tau \leq \log(t-\NN-1) $ since $\tildet\geq \NN+3$. Dividing both sides of the above equation by $t(t+1)$, and using the fact that $x_{t}\leq 2(\Bar{p} -\underline{p})^{2}$ for any $t$, we have
 \begin{align*}
     x_{t+1}
     ~\leq ~ & \left(\frac{2(\tildet-1) \tildet(\Bar{p} -\underline{p})^{2}}{t} + \V \right)\cdot \frac{1}{t+1} + \frac{1 + \max\{\theta_{1}, \theta_{2}\} \const\left(1  + \frac{(\NN+1)\log(t-\NN-1)}{t}\right)}{t+1}\\
      ~\overset{(a)}{\leq}~ &  \frac{2\tildet (\Bar{p} -\underline{p})^{2} + \V + 1 + \max\{\theta_{1}, \theta_{2}\} \const\left(1  + \frac{(\NN+1)\log(t-\NN-1)}{t}\right)}{t+1}\\
       ~\overset{(b)}{\leq}~ &  \frac{2\tildet (\Bar{p} -\underline{p})^{2} + \V + 1 + \ourrho \const}{t+1}\\
        ~\overset{(c)}{=}~ &  \frac{\const}{t+1}
      \,.
 \end{align*}
Here, (a) follows from the fact that $t\geq \tildet$, so $\frac{2(\tildet-1) \tildet(\Bar{p} -\underline{p})^{2}}{t} + \V\leq 2\tildet (\Bar{p} -\underline{p})^{2} + \V$; (b) follows from $t\geq \tildet\geq \trho$ so that $\frac{(\NN+1)\log(t-\NN-1)}{t} \leq  \frac{\ourrho}{\max\{\theta_{1}, \theta_{2}\}}-1$ according to Equation (\ref{def:OMD:Trho}); finally, (c) follows from Equation (\ref{eq:OMD:tildeT2}).

\subsection{Proof of Theorem \ref{thm:OMD:convergenceSuf}}
Here, we first provide a roadmap for the proof. Evoking Corollary \ref{cor:SNEBregmanUpdate}, we get
\begin{align}\label{eq:DR}
     \BDi(p_i^*, p_{i,t+1}) ~\leq~
       \BDi(p_i^*, p_{i,t})  - \epsilon_{i,t}\left(g_i^{*} - g_{i,t}\right)\left( p_i^{*} - p_{i,t} \right) +  \frac{\left(\epsilon_{i,t}\right)^2 \left(g_i^{*} - g_{i,t}\right)^2}{2\sigma_i} \,.
 \end{align}
By bounding the first order term $\left(g_i^{*} - g_{i,t}\right)\left( p_i^{*} - p_{i,t} \right)$ and the second order term $\frac{\left(g_i^{*} - g_{i,t}\right)^2}{2\sigma_i} $, we achieve a recursive relation in the form of 
\begin{align*}
    & \sum_{i =1,2,n} \BDi(p_i^*, p_{i,t+1}) ~\leq~  \sum_{i =1,2,n} \BDi(p_i^*, p_{i,t}) + \sum_{i =1,2,n} \kappa_{i,t}x_{i,t}\,,
\end{align*}
where we recall the definition $x_{i,t} = (p_{i}^{*} - p_{i,t})^{2}$ for $i = 1,2,n$ as in the proof of Theorem \ref{thm:OMD:rateDecStep}, and $\kappa_{i,t}$ is some constant that takes negative values if the conditions in the theorem's statement are satisfied. We then argue if  
$(\bm{p}_{t},r_{t})$ does not converge to the SNE, $\sum_{i =1,2,n} \BDi(p_i^*, p_{i,t}) $ will be greater than some positive constant $\epsilon > 0$ for all large enough $t$. Combining this with the above recursive relationship, this further implies that the distance between the price profile $(\bm{p}_{t},r_{t})$ and the SNE decreases by a positive constant for each period. This will eventually contradict the fact that Bregman divergence is positive.

We start our proof by recalling Equation (\ref{OMD:gradient}) which states $ g_{i}(\bm{p},r) = 2\beta_i p_{i} - \left(\alpha_{i} + \delta_{i} p_{-i} + \gamma_{i} r \right)$. Hence, 
\begin{align*}
  g_{i}^*- g_{i,t} 
    ~=~ & 2\beta_{i} \left( p_{i}^{*} - p_{i,t}\right)-  \delta_{i}\left(p_{-i}^{*} -  p_{-i,t}\right) - \gamma_{i}\left(r^{*} - r_{t} \right) \,.
\end{align*}
Furthermore, for $i=1,2$, we have
\begin{align}
\label{eq:OMD:squaredGrad}
    \left(g_{i}^*- g_{i,t}\right)^{2} ~\leq~ 8\beta_{i}^{2}x_{i,t} + 4\delta_{i}^{2}x_{-i,t}+ 4\gamma_{i}^{2}x_{n,t}\,,
\end{align}
where we used $(A + B +C)^{2} \leq 2A^{2} + 2(B+C)^{2} \leq 2A^{2} + 4B^{2}+4C^{2} $ for $A = 2\beta_{i} \left( p_{i}^{*} - p_{i,t}\right)$, $B = \delta_{i} \left( p_{-i}^{*} - p_{-i,t}\right)$, and $C = \gamma_{i}\left(r^{*} - r_{t}\right)$.  Hence, we have 
\begin{align}
\label{eq:OMD:convergenceSuf1}
     &\BDi(p_{i}^*, p_{i,t+1}) \nonumber \\
     ~\leq~ & \BDi(p_{i}^{*}, p_{i,t}) -  \epsilon_{i,t} \left( 2\beta_{i} x_{i,t} - 
    \delta_{i} \left(p_{-i}^{*} -  p_{-i,t}\right) \left(p_{i}^{*} - p_{i,t} \right) - \gamma_{i} \left(r^{*} - r_{t} \right) \left(p_{i}^{*} - p_{i,t} \right)
    \right)  + \frac{\left(\epsilon_{i,t}g_{i,t}\right)^2}{2\sigma_{i}} \nonumber \\ 
      ~\overset{(a)}{\leq}~ &  \BDi(p_{i}^{*}, p_{i,t}) -  \epsilon_{i,t} \left( 2\beta_{i} x_{i,t} - 
    \frac{\delta_{i}}{2} \left(x_{i,t} + x_{-i,t}\right) - \frac{\gamma_{i}}{2}\left(x_{n,t}  + x_{i,t}\right)\right)  + \frac{\left(\epsilon_{i,t}g_{i,t}\right)^2}{2\sigma_{i}}  \nonumber \\ 
    ~=~ &  \BDi(p_{i}^{*}, p_{i,t}) -  \epsilon_{1,t} \left( \frac{4\beta_{i} - \delta_{i} - \gamma_{i}}{2}x_{i,t}- 
    \frac{\delta_{i}}{2} x_{-i,t}- \frac{\gamma_{i}}{2}x_{n,t}\right) + \frac{\left(\epsilon_{i,t}g_{i,t}\right)^2}{2\sigma_{i}}  \nonumber\\
    ~\leq~ & \BDi(p_{i}^{*}, p_{i,t}) -   \left( \frac{\left(4\beta_{i} - \delta_{i} - \gamma_{i}\right)\epsilon_{i,t}}{2} - \frac{4\beta_{i}^{2}\left(\epsilon_{i,t}\right)^{2}}{\sigma_{i}}\right)x_{i,t}\nonumber\\
    &~~~~ +\left( \frac{\delta_{i}\epsilon_{i,t}}{2}  +\frac{2\delta_{i}^{2}\left(\epsilon_{i,t}\right)^{2}}{\sigma_{i}}  \right) x_{-i,t} + \left(\frac{\gamma_{i}\epsilon_{i,t}}{2}  + \frac{2\gamma_{i}^{2}\left(\epsilon_{i,t}\right)^{2}}{\sigma_{i}} \right)x_{n,t} \,.
\end{align}
In the  inequality (a),  we used the basic inequality $AB \leq (A^{2}+B^{2})/2$ twice, and the last inequality is obtained by invoking Equation (\ref{eq:OMD:squaredGrad}). 
Furthermore, we have $g_n^{*}- g_{n,t} ~=~  r^{*} - r_{t} - \left( \theta_1 \left(p_1^{*} - p_{1,t}\right)  + \theta_2 \left(p_2^{*} - p_{2,t} \right)\right)$. Thus,
\begin{align*}
   \left(g_n^{*}- g_{n,t} \right)^{2}~\leq~&  \frac{1}{2}x_{n,t} + \frac{1}{2}\left( \theta_1 \left(p_1^{*} - p_{1,t}\right)  + \theta_2 \left(p_2^{*} - p_{2,t} \right)\right)^{2}\\
    ~\overset{(a)}{\leq}~ & \frac{1}{2}x_{n,t} + \frac{1}{2}\left(\theta_1 x_{1,t}  + \theta_2 x_{2,t}\right)
   \,,
\end{align*}
where (a) follows from $\theta_{1} + \theta_{2}= 1$ and convexity. By applying the above inequality in Equation \eqref{eq:DR} with $i=n$, we have
\begin{align}
\label{eq:OMD:convergenceSuf3}
     & \BDn(p_{n}^{*}, p_{n,t+1})\nonumber\\
    ~{\leq}~ & \BDn(p_n^*, p_{n,t}) -  (1-a) \left( \frac{1}{2}x_{n,t}- 
    \frac{\theta_{1}}{2}x_{1,t} - \frac{\theta_{2}}{2}x_{2,t}\right)+ \frac{\left((1-a)g_{n,t}\right)^2}{2} \nonumber \\
    ~\leq ~ & \BDn(p_n^*, p_{n,t}) - \left(\frac{1-a}{2} - \frac{(1-a)^{2}}{4} \right)x_{n,t}    \nonumber \\
   &~~~~ + 
   \left(\frac{(1-a)\theta_1}{2}+ \frac{(1-a)^{2}\theta_{1}}{4}\right) x_{1,t} +    \left(\frac{(1-a)\theta_{2}}{2} + \frac{(1-a)^{2}\theta_{2}}{4}\right)x_{2,t} 
    \,,
\end{align}
where in the second inequality, we again use the inequality  $g_{n,t}^2\le \frac{1}{2}x_{n,t} + \frac{1}{2}\left(\theta_1 x_{1,t}  + \theta_2 x_{2,t}\right)$.

Summing up Equations (\ref{eq:OMD:convergenceSuf1}) (over $i=1,2$) and (\ref{eq:OMD:convergenceSuf3}), and collecting terms yields 
\begin{align}
\label{eq:OMD:decayingBregman}
    & \sum_{i =1,2,n} \BDi(p_i^*, p_{i,t+1}) ~\leq~  \sum_{i =1,2,n} \BDi(p_i^*, p_{i,t}) + \sum_{i =1,2,n} \kappa_{i,t}x_{i,t}\,,
\end{align}{}
where the coefficient for $x_{i,t}$ is
\begin{align}
\label{eq:OMD:decayingBregmanStepsize}
\kappa_{i,t} = \begin{cases}
    -   \frac{\left(4\beta_i - \delta_i - \gamma_i\right)\epsilon_{i,t}}{2} + \frac{4\beta_i^{2}\epsilon_{i,t}^{2}}{\sigma_{i}}+  \frac{\delta_{-i}\epsilon_{-i,t}}{2} + \frac{2\delta_{-i}^{2}\epsilon_{-i,t}^{2}}{\sigma_{-i}} + \frac{(1-a)\theta_{i}}{2}+ \frac{(1-a)^{2}\theta_{i}}{4}, &  i = 1,2
        \vspace{0.3cm}\\ 
    - \frac{1-a}{2}+ \frac{(1-a)^{2}}{4} +  \frac{\gamma_1\epsilon_{1,t}}{2}+ \frac{2\gamma_1^{2}\epsilon_{1,t}^{2}}{\sigma_1}  + \frac{\gamma_2 \epsilon_{2,t}}{2}  + \frac{2\gamma_2^{2}\epsilon_{2,t}^{2}}{\sigma_2}, &  i = n
    \end{cases}
\end{align}

Now, for $i=1, 2$, consider taking step size $\epsilon_{i,t} = \frac{\var \sigma_{i}}{\beta_{i}}(1-a)$, for some constant $\var > 0$ that will be determined later, and denote the corresponding $\kappa_{i,t}$ as $\kappa_{i}(z)$ (we drop the dependence on time $t$ as step sizes are constant), where for $i=1, 2$, 
\begin{align}
\label{eq:OMD:constKappaUB1}
    \kappa_{i}(\var) ~\overset{(a)}{=}~ & - \frac{\left(4\beta_i - \delta_i - \gamma_i\right)\var}{2\beta_{i}}(1-a)\sigma_{i}+ 4\var^{2}(1-a)^{2}\sigma_{i}+\frac{\var\delta_{-i}}{2\beta_{-i}}(1-a)\sigma_{-i}  \nonumber \\
    & ~~~~ +  \frac{2\delta_{-i}^{2}\var^{2}}{\beta_{-i}^{2}}(1-a)^{2}\sigma_{-i} + \frac{(1-a)\theta_{i}}{2}+ \frac{(1-a)^{2}\theta_{i}}{4} \nonumber\\
      ~\overset{(b)}{\leq}~ & - \frac{\left(4 -\frac{1}{\margin}\right)\var}{2}(1-a)\sigma_{i}+ 4\var^{2}(1-a)\sigma_{i}+\frac{\var}{2\margin}(1-a)\sigma_{-i}  +  \frac{2\var^{2}}{\margin^{2}}(1-a)\sigma_{-i} + \frac{3(1-a)}{4} \nonumber\\
       ~=~ & (1-a)\left(\left(4\sigma_{i} +\frac{2\sigma_{-i}}{\margin^{2}}\right)\var^{2} - \left( \left(2- \frac{1}{2\margin}\right)\sigma_{i} - \frac{\sigma_{-i}}{2\margin}\right)\var + \frac{3}{4}\right) \nonumber\\
      ~\defeq~ & (1-a)f_{i,\margin}(\var)
      \,.
\end{align}
Here, in (a) we substitute $\epsilon_{i,t} = \frac{\var \sigma_{i}}{\beta_{i}}(1-a)$ for $i = 1,2$; in (b)
we use the fact that $\theta_{i}, a \in (0,1)$ (wich implies $(1-a)^{2} \leq 1-a$) and $\beta_{i} > \margin(\delta_{i} + \gamma_{i}) > \margin\max\{\delta_{i}, \gamma_{i}\}$.

We follow a similar argument as above and obtain 
\begin{align}
\label{eq:OMD:constKappaUB2}
    \kappa_{n}(\var) ~\overset{(a)}{=}~ &   - \frac{1-a}{2}+ \frac{(1-a)^{2}}{4}  + \frac{\var\gamma_{1}\sigma_{1}}{2\beta_{1}}(1-a)+ \frac{2\var^{2}\gamma_{1}^{2}\sigma_{1}}{\beta_{1}^{2}}(1-a)^{2}\nonumber \\
    & ~~~~ + \frac{\var\gamma_{2}\sigma_{2}}{2\beta_{2}}(1-a)+ \frac{2\var^{2}\gamma_{2}^{2}\sigma_{2}}{\beta_{2}^{2}}(1-a)^{2} \nonumber  \\
     ~\overset{(b)}{\leq}~ & (1-a)\left(-\frac{1}{4} + \frac{\var\sigma_{1}}{2\margin}+ \frac{2\var^{2}\sigma_{1}}{\margin^{2}}+ \frac{\var\sigma_{2}}{2\margin}+ \frac{2\var^{2}\sigma_{2}}{\margin^{2}}\right) \nonumber  \\
      ~=~ &(1-a)\left(\frac{2}{\margin^{2}}\left(\sigma_{1} + \sigma_{2}\right)\var^{2}+ \frac{1}{2\margin}  \left(\sigma_{1} + \sigma_{2}\right)\var -\frac{1}{4}\right) \nonumber\\
      ~\defeq~ & (1-a)f_{n,\margin}(\var)
      \,,
\end{align}
where in (a) we substitute $\epsilon_{i,t} = \frac{\var \sigma_{i}}{\beta_{i}}(1-a)$ for $i = 1,2$; in (b) we used the fact that $\theta_{i}, a \in (0,1)$ and $\beta_{i} > \margin(\delta_{i} + \gamma_{i}) > \margin \max\{\delta_{i}, \gamma_{i} \}$ for any $i=1,2$.

Now, recall the definition $\mathcal{S}_{i,\margin} = \left\{\var > 0: f_{i,\margin}(\var) < 0 \right\}$. Then,  if we have $\cap_{i=1,2,n} \mathcal{S}_{i,m}  \neq \emptyset$, taking any $\cone \in \cap_{i=1,2,n} \mathcal{S}_{i,\margin} $ yields $\kappa_{i}(\cone)< 0$ for $i = 1,2,n$. Hence, Equation (\ref{eq:OMD:decayingBregman}) now becomes 
\begin{align}
\label{eq:OMD:decayingBregman1}
    & \sum_{i =1,2,n} \BDi(p_i^*, p_{i,t+1}) ~\leq~  \sum_{i =1,2,n} \BDi(p_i^*, p_{i,t}) + \sum_{i =1,2,n} \kappa_{i}(\cone)x_{i,t}, \quad \kappa_{i}(\cone) < 0\,.
\end{align}
Therefore, we know that 
\begin{align}
\label{eq:OMD:constRecurrence}
     \sum_{i =1,2,n} \BDi(p_i^*, p_{i,t+1}) ~<~  \sum_{i =1,2,n} \BDi(p_i^*, p_{i,t})\,.
\end{align}
Furthermore, by strong convexity, 
$$  \sum_{i =1,2,n} \BDi(p_i^*, p_{i,t}) \geq  \sum_{i =1,2,n} \frac{\sigma_{i}^{2}}{2} \left(p_i^*-  p_{i,t}\right)^2 \geq \frac{\min_{i=1,2,n} \sigma_{i}^{2}}{2} \norm{\bm{p}^{*} - \bm{p}_{t}}.$$
Hence, for any small $\epsilon > 0$, if there exists some $\te \in \N^{+}$  such that $\sum_{i =1,2,n} \BDi(p_i^*, p_{i,\te })\leq \frac{\epsilon\cdot \min_{i=1,2,n} \sigma_{i}^{2}}{2}$, then by Equation (\ref{eq:OMD:constRecurrence}),  $\sum_{i =1,2,n} \BDi(p_i^*, p_{i,t})\leq \frac{\epsilon\cdot \min_{i=1,2,n} \sigma_{i}^{2}}{2}$ for all $t \geq \te$, which further implies $\norm{\bm{p}^{*} - \bm{p}_{t}}\leq \epsilon$ for all  $t \geq \te$. Hence $\left(\bm{p}_{t}, r_{t}\right) \overset{t\to \infty}{\longrightarrow} \left(\bm{p}^{*}, r^{*}\right)$.

Thus, it remains to show that for any small $\epsilon > 0$, there exists $\te > 0$ such that $\sum_{i =1,2,n} \BDi(p_i^*, p_{i, \te}) < \epsilon$. We will prove this by contradiction. If this is not the case, there exists $\epsilon > 0$, and $\sum_{i =1,2,n} \BDi(p_i^*, p_{i,t}) \geq \epsilon $ for all $t \geq 0$. Define $R(\var_{1},\var_{2}, \var_{3}) = \sum_{i=1,2,3}R_{i}(\var_{i})$ for any $\var_{1},\var_{2}, \var_{3}\in \R$, and slightly abuse the notation to define $D: \R^{3} \times \R^{3} \to \R $ as the Bregman divergence with respect to $R$. In the rest of this proof for simplicity we also write $\bm{p}^{*} = (p_{1}^{*}, p_{2}^{*}, r^{*})$ and $\bm{p}_{t} =  (p_{1,t}, p_{2,t}, r_{t})$. A simple analysis shows $D(\bm{p}^{*}, \bm{p}_{t}) =\sum_{i =1,2,n} \BDi(p_i^*, p_{i,t})$. 

Since $R_i$ is continuously differentiable (by definition of Bregman divergence), $R$ is also continuously differentiable, and hence it is easy to see for any $\bm{x},\bm{y} \in \R^{3}$ there exists $ \delta > 0$ such that 
$$ D(\bm{x},\bm{y} ) < \epsilon, \quad \forall \norm{\bm{x}-\bm{y}} < \delta.$$
Since we assumed  $D(\bm{p}^{*}, \bm{p}_{t}) = \sum_{i =1,2,n} \BDi(p_i^*, p_{i,t}) \geq \epsilon $ for all $t\geq 0$, the above implies 
$\norm{\bm{p}^{*} - \bm{p}_{t}} \geq \delta$ for all $t \geq 0$. Hence following Equation (\ref{eq:OMD:decayingBregman1}),
\begin{align*}
    D(\bm{p}^{*}, \bm{p}_{t+1})  ~\leq~ & D(\bm{p}^{*}, \bm{p}_{t}) + \sum_{i =1,2,n} \kappa_{i}(\cone)\left(p_{i}^{*} - p_{i,t} \right)^{2} \\
    ~\leq~ &  D(\bm{p}^{*}, \bm{p}_{t}) +\max_{i=1,2,n}\kappa_{i}(\cone) \sum_{i =1,2,n} \left(p_{i}^{*} - p_{i,t} \right)^{2}\\
      ~=~ &  D(\bm{p}^{*}, \bm{p}_{t})+ \max_{i=1,2,n}\kappa_{i}(\cone)  \cdot \norm{\bm{p}^{*} - \bm{p}_{t}}^{2} \\
       ~\overset{(a)}{\leq}~ &  D(\bm{p}^{*}, \bm{p}_{t})+ \delta^{2}\max_{i=1,2,n}\kappa_{i}(\cone)  \\
       ~\overset{(b)}{\leq}~ &  D(\bm{p}^{*}, \bm{p}_{1}) + t \delta^{2}\max_{i=1,2,n}\kappa_{i}(\cone) \,,
\end{align*}
where (a) follows because $\kappa_{i}(\cone) < 0$ for $i = 1,2,n $  and $\norm{\bm{p}^{*} - \bm{p}_{t}} \geq \delta$ for all $t \geq 0$; (b) follows from a telescoping argument.
Finally, $\max_{i=1,2,n}\kappa_{i}(\cone) < 0$ implies the right hand side in the above inequality goes to negative infinity as  $t$ goes to infinity. This implies that  $ D(\bm{p}^{*}, \bm{p}_{t+1})  = \sum_{i =1,2,n} \BDi(p_i^*, p_{i,t}) \leq - \infty $, which contradicts nonnegativity of Bregman divergence. Hence, for any small $\epsilon > 0$, there exists $\te > 0$ such that $\sum_{i =1,2,n} \BDi(p_i^*, p_{i,\te}) < \epsilon$,  concluding the proof.

\subsection{Proof of Corollary \ref{cor:OMD:convergenceSuf}}
When $\sigma_{1} = \sigma_{2} = \sigma$, Equation (\ref{eq:OMD:constConvRegion}) becomes 
\begin{align*}
    f_{i,\margin}(\var) = \begin{cases}
    2\sigma \left(2 +\frac{1}{\margin^{2}}\right)\var^{2} - \sigma\left(2 -\frac{1}{\margin}\right)\var+ \frac{3}{4} & i=1,2\\
    \frac{4\sigma}{\margin^{2}}\var^{2}+ \frac{\sigma}{\margin}\var -\frac{1}{4} & i = n
    \end{cases}\,.
\end{align*}
Since in this case the function $f_{1,\margin}(\var)$ and  $f_{2,\margin}(\var)$ are identical, we only consider $f_{1,\margin}(\var)$. Note that the function $f_{1,\margin}(\var)$ has two distinct zero roots if and only if its discriminant is strictly greater than 0, i.e., $\left(1-\frac{1}{2\margin}\right)^{2} - \frac{3}{2\sigma}\left(2+\frac{1}{\margin^{2}}\right) > 0$ which is equivalent to $\sigma > \frac{6(2\margin^{2} + 1)}{(2\margin-1)^{2}}$. Therefore, when $f_{1,\margin} $ has two distinct zero roots, the smaller one is given by 
\begin{align}
\label{eq:const:m1}
    \var_{1} ~=~ & \frac{1-\frac{1}{2\margin} - \sqrt{\left(1-\frac{1}{2\margin}\right)^{2} - \frac{3}{2\sigma}\left(2+\frac{1}{\margin^{2}}\right)}}{2\left(2+\frac{1}{\margin^{2}} \right)} \nonumber \\
    ~=~ & \frac{1}{\sigma}\cdot \underbrace{\frac{{3}/{4}}{1-\frac{1}{2\margin} + \sqrt{\left(1-\frac{1}{2\margin}\right)^{2} - \frac{3}{2\sigma}\left(2+\frac{1}{\margin^{2}}\right)}}}_{A}> 0.
\end{align}
Similarly, the discriminant of $f_{n,\margin}$ (i.e., $\frac{1}{\margin^{2}}+ \frac{4}{\sigma \margin^{2}}$) is always positive, so  $f_{n,\margin}$ always has two zero roots. The larger one is given by 
\begin{align}
    \label{eq:const:m2}
    \var_{2} = \frac{-\frac{1}{\margin} + \sqrt{\frac{1}{\margin^{2}}+ \frac{4}{\sigma \margin^{2}}}}{\frac{8}{\margin^{2}}}= \frac{1}{\sigma} \cdot \underbrace{\frac{{1}/{2}}{\frac{1}{\margin} + \sqrt{\frac{1}{\margin^{2}}+ \frac{4}{\sigma \margin^{2}}}}}_{B} > 0.
\end{align}
For any $\sigma > \frac{6(2\margin^{2} + 1)}{(2\margin-1)^{2}}$, the two roots of $f_{1,\margin}(\var)$ are both positive, while  $f_{n,\margin}(\var)$ always has one positive root and one negative root. Hence, using a simple geometric argument regarding two quadratic functions, it is easy to see that if $\var_{2} > \var_{1}$, any $\cone \in (\var_{1},\var_{2})$ satisfies $f_{1,\margin}(\cone) , f_{n,\margin}(\cone) < 0$.

Now, consider $ \var_{2} -\var_{1} = \frac{1}{\sigma} (B- A)$. Since we observe $A$ is decreasing in $\sigma$ and $B$ is increasing in $\sigma$, we have $B-A$ is increasing in $\sigma$. By direct calculations, we see that when $\sigma =  \frac{\left(2\margin^{2} + 7\right)^{2}}{8\margin^{3} - 36\margin + 8} > 0$, $\var_{1} = \var_{2}$ (i.e., $B -A =0$). Therefore because $B-A$ is increasing in $\sigma$, we conclude $B - A>  0$ for any $$\sigma > \sigma_{0}\defeq \max\left\{\frac{6(2\margin^{2} + 1)}{(2\margin-1)^{2}}, \frac{\left(2\margin^{2} + 7\right)^{2}}{8\margin^{3} - 36\margin + 8}\right\}\,, $$
which implies $\var_{2} > \var_{1}$ for any $\sigma > \sigma_{0}$. 

In sum, we conclude for any $\margin > 2$, if $\sigma >\sigma_0$, there exists $\cone > 0$ that depends on $\sigma$ and $\margin$ such that $f_{i,\margin}(\cone) < 0$ for $i=1,2,n$, and by Theorem \ref{thm:OMD:convergenceSuf}, this implies that there exist constant step sizes under which prices and reference prices converge to the unique  interior SNE.

\subsection{Proof of Theorem \ref{thm:OMD:constExpConvergence}}
When $R_{i}(x) = \frac{\sigma}{2} x^{2}$ for $i = 1,2$, we have $\sigma_{1} = \sigma_{2} = \sigma$, and Equation (\ref{eq:OMD:constConvRegion}) becomes
\begin{align*}
    f_{i,\margin}(\var) = \begin{cases}
    2\sigma \left(2 +\frac{1}{\margin^{2}}\right)\var^{2} - \sigma\left(2 -\frac{1}{\margin}\right)\var + \frac{3}{4}, & i=1,2\\
    \frac{4\sigma}{\margin^{2}}\var^{2}+ \frac{\sigma}{\margin}\var -\frac{1}{4}, & i = n    \end{cases}\,.
\end{align*}
We  define $ h_{i,\margin}(\var):=  f_{i,\margin}(\var)/2\sigma$ for $i=1,2$ and $h_{n,\margin}(\var):= f_{n,\margin}(\var)$, i.e. 
\begin{align}
\begin{aligned}
\label{eq:OMD:constConvRegion1}
    h_{i,\margin}(\var) = \begin{cases}
    \left(2 +\frac{1}{\margin^{2}}\right)\var^{2} - \left(1 -\frac{1}{2\margin}\right)\var + \frac{3}{8\sigma} & i=1,2\\
    \frac{4\sigma}{\margin^{2}}\var^{2}+ \frac{\sigma}{\margin}\var -\frac{1}{4} & i = n
    \end{cases}\,.
\end{aligned}
\end{align}
Note that for any $i=1,2,n$, $f_{i,\margin}(\var) < 0$ if and only if $h_{i,\margin}(\var) < 0$. Hence, according to Corollary \ref{cor:OMD:convergenceSuf}, we know that when $\margin > 2$ and $\sigma > \sigma_{0} = \max\left\{\frac{6(2\margin^{2} + 1)}{(2\margin-1)^{2}}, \frac{\left(2\margin^{2} + 7\right)^{2}}{8\margin^{3} - 36\margin + 8} \right\}$, for any $M \in (\var_{1},\var_{2})$ (defined in Equations (\ref{eq:const:m1}) and (\ref{eq:const:m2})) we have $f_{i,\margin}(\cone) < 0$ for $i=1,2,n$, which implies  $h_{i,\margin}(\cone) < 0$ for $i=1,2,n$. Furthermore, via a simple geometric argument, the quadratic functions  $h_{1,\margin}$ (with two positive zero roots) and $h_{n,\margin}$ (with two zero roots, one positive and one negative) have a unique intersection point $\widetilde{\cone} \in (\var_{1},\var_{2})$. Define 
 $H\defeq h_{1,\margin}(\widetilde{\cone}) = h_{2,\margin}(\widetilde{\cone}) = h_{n,\margin}(\widetilde{\cone}) <0$. Furthermore, since $\min_{m\geq 0}h_{n,\margin}(\var) = -\frac{1}{4}$, we have   
 $$-\frac{1}{4} \leq H = h_{1,\margin}(\widetilde{\cone}) = h_{2,\margin}(\widetilde{\cone}) = h_{n,\margin}(\widetilde{\cone}) <0\,. $$
 
Now, note that when $R_{1}(x) = R_{2}(x) =\frac{\sigma}{2} x^{2}$, $\BDone(p,p')= \BDtwo(p,p') = \frac{\sigma}{2}(p-p')^{2}$. Also recall $R_{n}(x) = \frac{1}{2}x^{2}$, so $\BDn(p,p') = \frac{1}{2}(p-p')^{2}$. Hence, $\sum_{i =1,2,n} \BDi(p_i^*, p_{i,t}) =  \frac{1}{2} \left(\sigma x_{t}  + x_{n,t}\right)$, where we define $x_{t} = \norm{\bm{p}^{*} - \bm{p}_{t}}^{2}$ and $x_{n,t}= (r^{*} - r_{t})^{2}$ as in the proof of Theorem \ref{thm:OMD:convergenceSuf}. Hence, by taking $\epsilon_{i,t} =\frac{\sigma \widetilde{\cone}(1-a)}{\beta_{i}}$, and continuing from Equation (\ref{eq:OMD:decayingBregman1}), we get 
\begin{align*}
  \frac{1}{2} \left(\sigma x_{t+1}  + x_{n,t+1}\right) ~\leq~ & \frac{1}{2} \left(\sigma x_{t}  + x_{n,t}\right)+ \sum_{i =1,2,n} \kappa_{i}(\widetilde{\cone})x_{i,t} \nonumber \\
    ~\overset{(a)}{\leq}~ & \frac{1}{2} \left(\sigma x_{t}  + x_{n,t}\right) + (1-a)\sum_{i =1,2,n} f_{i,\margin}(\widetilde{\cone})x_{i,t}\\
     ~\overset{(b)}{=}~ & \frac{1}{2} \left(\sigma x_{t}  + x_{n,t}\right)  + (1-a)\left(\sum_{i =1,2} \sigma h_{i,\margin}(\widetilde{\cone})x_{i,t} + h_{n,\margin}(\widetilde{\cone})x_{n,t}\right)\\
      ~\overset{(c)}{=}~ & \frac{1}{2} \left(\sigma x_{t}  + x_{n,t}\right) + (1-a) H\cdot\left(\sigma x_{t} + x_{n,t}\right) \\
      ~=~ & \frac{1}{2} \left(1 + 2(1-a)H\right)\left(\sigma x_{t}  + x_{n,t}\right)\,.
\end{align*}
Here, (a) follows from upper bounding $\kappa_{i}(\var)$ with $f_{i,\margin}(\var)$ for any $\var > 0$ and $i=1,2,n$ in Equations (\ref{eq:OMD:constKappaUB1}) and  (\ref{eq:OMD:constKappaUB2}) within the proof of Theorem \ref{thm:OMD:convergenceSuf}; (b) follows from the definition of $h_{i,\margin}$ in Equation (\ref{eq:OMD:constConvRegion1}); 
(c) follows from the definition of $H \defeq h_{i,\margin}(\widetilde{\cone})\in [-\frac{1}{4},0)$ for $i = 1,2,n$ and $x_{t} = x_{1,t} + x_{2,t}$.

Using a telescoping argument, we have
$$ \sigma x_{t} < \sigma x_{t}  + x_{n,t} \leq \left(1+ 2(1-a)H\right)^{t}\left(\sigma x_{1}+ x_{n,1}\right) \leq  \left(\sigma x_{1}+ x_{n,1}\right) \left(\frac{1+a}{2}\right)^{t}\,,
$$ 
where the final inequality follows from $0 < 1 + 2(1-a)H \leq \frac{1+a}{2}$ since $H \in [-\frac{1}{4},0)$. Finally, because
$x_{1} \leq 2\left(\Bar{p} - \underline{p}\right)^{2}$ and $x_{n,1}\leq \left(\Bar{p} - \underline{p}\right)^{2}$,
we have 
$$  x_{t} < \left( x_{1}+ \frac{1}{\sigma} x_{n,1}\right) \left(\frac{1+a}{2}\right)^{t} \leq \frac{1+2\sigma}{\sigma} \left(\Bar{p} - \underline{p}\right)^{2} \left(\frac{1+a}{2}\right)^{t}\,.
$$ 

\subsection{Supplementary Figures for Section \ref{sec:firstOrder}}
\label{app:sec:addFigures}
\begin{figure}[H]
	\centering
	\begin{subfigure}{0.45\linewidth}
		\centering
		\includegraphics[width=1\linewidth]{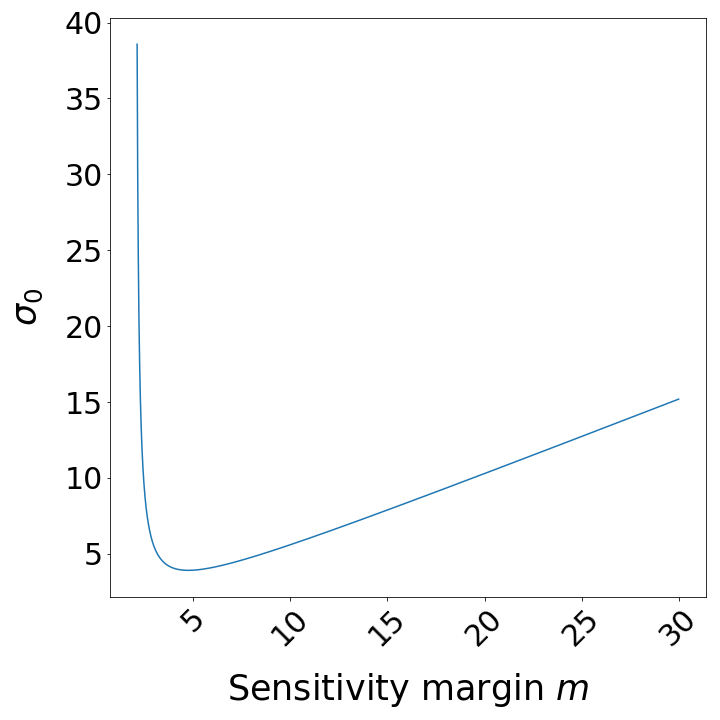}
    \caption{}
		\label{fig:sig0K}
	\end{subfigure}
		\begin{subfigure}{0.45\linewidth}
		\centering
		\includegraphics[width=1.0\linewidth]{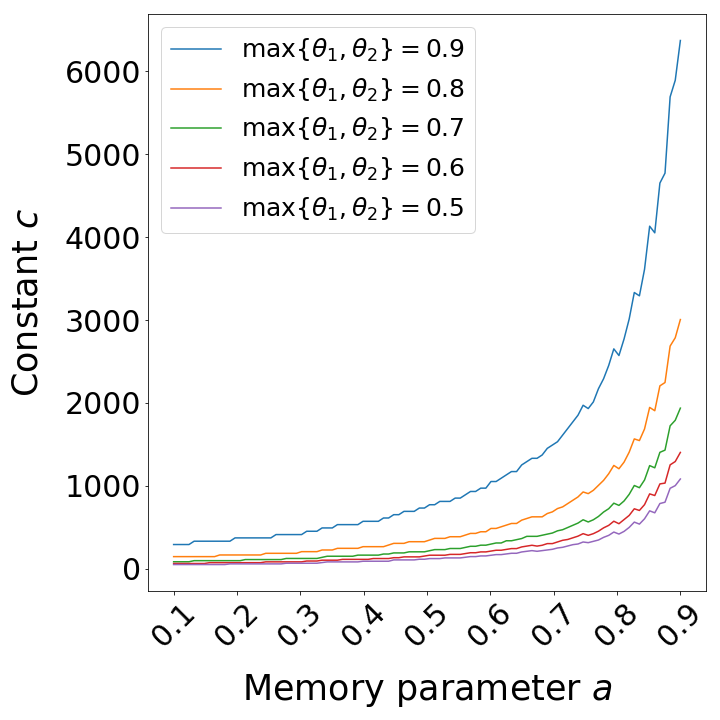}
        \caption{}
	\label{fig:rate_constant}
	\end{subfigure}
	\caption{(a) $\sigma_{0}$ as a function of sensitivity margin $m$, where $\sigma_{0}$ is defined in Corollary \ref{cor:OMD:convergenceSuf} (b) Illustration of absolute constant $c$ in Theorem \ref{thm:OMD:rateDecStep} w.r.t. memory parameter $a$ and $\max\{ \theta_{1}, \theta_{2}\}$. All other model parameters take respective values as in Example \ref{ex:setting}, and firm $i=1,2$ again adopts regularizer $R_{i}(\var) = \frac{1}{2}\var^{2}$.} 
\end{figure}

\section{Supplementary Lemmas of Section \ref{sec:firstOrder}}

\begin{lemma}
\label{lem:OMD:bregmanupdateUB}
For $i = 1,2,n$ and any $\tilde{\var} \in \mathcal{P}$, we have for any $t \in \N^{+}$,
\begin{align}
\label{eq:OMD:distanceApprox}
    \BDi(\tilde{\var}, p_{i,t+1})~\leq~ \BDi(\tilde{\var}, p_{i,t}) +\epsilon_{i,t} \cdot g_{i,t} \left(\tilde{\var}- p_{i,t} \right)+  \frac{\left(\epsilon_{i,t}g_{i,t}\right)^2}{2\sigma_i}\,.
\end{align}
\end{lemma}
\begin{proof} 
In the projection step of Algorithm \ref{algo:firmOMD}, we have $p_{i,t+1} = \Pi_{\mathcal{P}}(y_{i,t+1})$. Since we are working with one-dimensional decision sets, it is easy to see that $\Pi_{\mathcal{P}}(y_{i,t+1}) = \arg\min_{p\in \mathcal{P}}\BDi(p, y_{i,t+1})$ due to convexity of $R_i$. Recalling the definition $R_i'(p) = \frac{d R_{i}(\var)}{d\var}\Big|_{\var = p}$, we have
\begin{align*}
    p_{i,t+1} =  \arg\min_{p\in \mathcal{P}}\BDi(p, y_{i,t+1}) ~=~ &  \arg\min_{p\in \mathcal{P}} R_{i}(p) - R_{i}(y_{i,t+1}) - R_{i}'(y_{i,t+1}) (p-y_{i,t+1})\\ 
    ~=~ &  \arg\min_{p\in \mathcal{P}} R_{i}(p)  - p\cdot R_{i}'(y_{i,t+1})\\
     ~\overset{(a)}{=}~ &  \arg\min_{p\in \mathcal{P}} R_{i}(p)  - p\cdot \left(R_{i}'(p_{i,t}) - \epsilon_{i,t}g_{i,t}\right)\\
     ~=~ &  \arg\min_{p\in \mathcal{P}}  R_{i}(p)  -  R_{i}(p_{i,t}) -   R_{i}'(p_{i,t})\left(p - p_{i,t}\right) + p\cdot \epsilon_{i,t}g_{i,t} \\
      ~=~ &  \arg\min_{p\in \mathcal{P}} \BDi(p, p_{i,t}) + p\cdot \epsilon_{i,t}g_{i,t} \,.
\end{align*}
Here (a) 
follows from the proxy update step in Algorithm \ref{algo:firmOMD}. Now, evoking Lemma \ref{lem:bregman} (ii) by taking $x = p$, $f(p) =  p \cdot \epsilon_{i,t}g_{i,t}$, $z = p_{i,t}$, $y= \tilde{\var}\in \mathcal{P}$, we have 
\begin{align*}
    \BDi(\tilde{\var}, p_{i,t+1}) \leq \BDi(\tilde{\var}, p_{i,t}) + \epsilon_{i,t}g_{i,t} \left(\tilde{\var}- p_{i,t+1} \right) - \BDi(p_{i,t}, p_{i,t+1})\,.
\end{align*}

It then follows that 
\begin{align*}
    & \BDi(\tilde{\var}, p_{i,t+1}) \nonumber \\
    ~\leq~ &  \BDi(\tilde{\var}, p_{i,t}) +\epsilon_{i,t}g_{i,t} \left( \tilde{\var} - p_{i,t} \right) +  \epsilon_{i,t}g_{i,t} \left( p_{i,t} - p_{i,t+1} \right) - \BDi(p_{i,t}, p_{i,t+1}) \nonumber\\
    ~\overset{(a)}{\leq}~ &  \BDi(\tilde{\var}, p_{i,t}) +\epsilon_{i,t}g_{i,t} \left( \tilde{\var} - p_{i,t} \right) +  \epsilon_{i,t}g_{i,t} \left( p_{i,t} - p_{i,t+1} \right) - \frac{\sigma_i}{2}\left(p_{i,t}- p_{i,t+1}\right)^2 \nonumber\\
    ~\leq~ &  \BDi(\tilde{\var}, p_{i,t}) +\epsilon_{i,t}g_{i,t} \left( \tilde{\var} - p_{i,t} \right) +  \frac{\left(\epsilon_{i,t}g_{i,t}\right)^2}{2\sigma_i}\,,
\end{align*}
where (a) follows from strong convexity of $R_i$.
\end{proof}

\begin{corollary}\label{cor:SNEBregmanUpdate}
Under Assumption 1, let $(\bm{p}^{*},r^{*})$ be the unique interior SNE as illustrated in Lemma \ref{lem:uniqueNE}, then for $i = 1,2,n$, 
\begin{align}
\label{eq:SNEFOC}
     g_i^* = \frac{\partial \widetilde{\pi}_i}{\partial p_i}\Big|_{\bm{p} = \bm{p^{*}}, r = r^{*}} = 0\,,
\end{align}
and for any $t \in \N^{+}$
\begin{align}
\label{eq:SNEBregmanUpdate}
    \BDi(p_{i}^{*}, p_{i,t+1}) 
    ~\leq~  \BDi(p_{i}^{*}, p_{i,t}) -\epsilon_{i,t}\left(g_{i}^{*} - g_{i,t} g_{i}^{*}\right) \left(p_{i}^{*} - p_{i,t} \right) +  \frac{\left(\epsilon_{i,t}g_{i,t}\right)^2}{2\sigma_i}\,.
\end{align}
\end{corollary}
\begin{proof}
Similar to the proof of Lemma \ref{lem:uniqueNE} and Proposition \ref{lem:OMD:PSNE}, the SNE  $(\bm{p}^{*},r^{*})$ must satisfy first order conditions w.r.t. quadratic cost function $\widetilde{\pi}_{1},\widetilde{\pi}_{2},\widetilde{\pi}_{n} $, respectively, due to the fact that it lies in the interior of the decision set. So $ g_i^*  = 0$ for $i = 1,2,n$.

Furthermore, Evoking Lemma \ref{lem:OMD:bregmanupdateUB} by replacing $\var$ with $p_{i}^{*}$ and combining  $ g_i^*  = 0$ yields the second part of the proof.
\end{proof}

\begin{lemma}[Lemma 3.1 and 3.2 of \cite{chen1993convergence}]
\label{lem:bregman}
Let $\BD: \mathcal{C} \times \mathcal{C} \to \R^+$ be the Bregman divergence associated with convex function $R$ on the convex set $\mathcal{C}$: $\BD(x,y) = R(x) - R(y) -  R'(y)(x- y) ~~~~,\forall x,y \in \mathcal{C}.$ Then, 
\begin{enumerate}
    \item[(i)]  For any $x,y,z \in \mathcal{C}$, $\BD(x,y)  + \BD(y,z)= \BD(x,z) +\left(R'(z) -  R'(y)\right)(x - y)\rangle $. 
    \item[(ii)] Let $f: \mathcal{C} \to \R$ be any convex function and $z \in \mathcal{C}$. If $x^* = \arg\min_{x\in\mathcal{C}}\left\{f(x) +  \BD(x,z) \right\}$, then for any $y\in \mathcal{C}$, we have $f(y) + \BD(y,z) \geq f(x^*) + \BD(x^*,z) +  \BD(y,x^*).$
\end{enumerate}
\end{lemma}
The proofs for the above lemma are very standard and we will omit them in this paper.

\begin{lemma}
\label{lem:boundSum}
Let $a\in(0,1)$,  $\NN = \left \lceil\frac{a}{1-a}\right \rceil + 1$, and $\MN = \left\lceil \frac{\frac{a}{1-a}(\NN+1)}{\NN - \frac{a}{1-a}}\right\rceil$. Then, for any ${t} \geq \NN + \MN$, we have
\begin{align*}
    \sum_{\tau=\NN + \MN}^{{t}} \frac{a^{-\tau}}{\tau} ~\leq~& \frac{1}{1-a}\cdot \frac{a^{-{t}}}{{t}-\NN}\,.
\end{align*}
\end{lemma}
\begin{proof} [Proof of Lemma \ref{lem:boundSum}]
We adopt an induction argument with  hypothesis  $\sum_{\tau=\NN+\MN}^{{t}} \frac{a^{-\tau}}{\tau} \leq \frac{1}{1-a}\cdot \frac{a^{-{t}}}{{t}-\NN}$. For the base case, consider ${t} = \NN+ \MN$. We can easily see $\frac{a^{-( \NN+ \MN)}}{ \NN+ \MN} <\frac{1}{1-a}\cdot \frac{a^{-( \NN + \MN)}}{\MN}$. Now assume that the induction hypothesis holds for some 
some ${t} \geq \NN + \MN$. We will show $\sum_{\tau=\NN+\MN}^{{t}+1} \frac{a^{-\tau}}{\tau} \leq \frac{1}{1-a}\cdot \frac{a^{-({t}+1)}}{{t}-\NN + 1}$. We start with 
\begin{align*}
    \sum_{\tau=\NN+\MN}^{{t}+1} \frac{a^{-\tau}}{\tau} ~\leq~& \frac{1}{1-a}\cdot \frac{a^{-{t}}}{{t}-\NN} + \frac{a^{-({t}+1)}}{{t}+1}
    ~=~ \frac{a^{-({t}+1)}}{1-a}\cdot \left(\frac{a}{{t}-\NN} + \frac{1-a}{{t}+1}\right)\,.
\end{align*}
Furthermore, 
\begin{align*}
    \frac{a}{{t}-\NN} + \frac{1-a}{{t}+1} - \frac{1}{{t}-\NN + 1} ~=~ & a\left(\frac{1}{{t}-\NN} -\frac{1}{{t}-\NN + 1} \right) + (1-a)\left( \frac{1}{{t}+1}-\frac{1}{{t}-\NN + 1}\right) \\
     ~=~ & \frac{1}{{t}-\NN+ 1}  \left( \frac{a}{{t}-\NN} -  \frac{(1-a)\NN}{{t}+1}\right) \\
      ~=~ & \frac{1}{{t}-\NN + 1} \cdot \frac{\left(a - \NN(1-a)\right){t} + (1-a)\NN^{2} + a}{({t}-\NN)({t}+1)} \\
    ~\overset{(a)}{\leq}~ & 0 \,,
\end{align*}{}
where (a) follows from $\NN = \left\lceil \frac{a}{1-a}\right\rceil + 1>  \frac{a}{1-a}$ and the fact that 
$$ \frac{(1-a)\NN^{2}+ a}{\NN(1-a) -a } = \frac{\NN^{2} + \frac{a}{1-a}}{\NN - \frac{a}{1-a}} = \NN + \frac{\frac{a}{1-a}(\NN+1)}{\NN - \frac{a}{1-a}} < \NN + \MN \leq {t}.$$
Therefore, we can conclude that
\begin{align*}
    \sum_{\tau=\NN+\MN}^{{t}+1} \frac{a^{-\tau}}{\tau} ~\leq~& \frac{1}{1-a}\cdot \frac{a^{-({t}+1)}}{{t}-\NN+1}\,,
\end{align*}
which is the desired result. 
\end{proof}

%
%
%




\end{document}